%% file: main.tex
\documentclass[lettersize,journal]{IEEEtran}

\usepackage{balance}
\usepackage{comment}

\usepackage{amsmath,mathtools} %
\usepackage{amsfonts,amssymb}

\usepackage{xspace} %

\usepackage{amsmath}

\usepackage{graphicx}
\graphicspath{{figs/}}

\usepackage{subcaption}

\usepackage{hyperref}
\hypersetup{
 colorlinks=true,
 linkcolor=blue,
 citecolor =blue,
 filecolor=black, 
 urlcolor=blue,
 }
\usepackage{todonotes}

\usepackage{enumitem}

\usepackage{pgfplots}\pgfplotsset{compat=1.9} 
\usepackage{grffile}
\pgfplotsset{compat=newest} %
\usetikzlibrary{plotmarks}
\usetikzlibrary{arrows.meta}
\usepgfplotslibrary{patchplots}

\usepackage{tabularx}
\usepackage{diagbox}
\newcolumntype{K}[1]{>{\centering\arraybackslash}p{#1}}

\usepackage[ruled]{algorithm}
\usepackage[noend]{algpseudocode}

\makeatletter
\newcommand{\multiline}[1]{%
 \begin{tabularx}{\dimexpr\linewidth-\ALG@thistlm}[t]{@{}X@{}}
 #1
 \end{tabularx}
}
\makeatother

\usepackage[per-mode=symbol]{siunitx}

\usepackage{cite} %
\usepackage{array} %
\usepackage{booktabs}  %

\usepackage{cite}

\newcommand{\CAV}[1]{CAV\textendash\ensuremath{#1}\xspace}
\newcommand{\HDV}[1]{HDV\textendash\ensuremath{#1}\xspace}

\newcommand{\bbsym}[1]{\ensuremath{\boldsymbol{#1}}}

\usepackage{amsthm}
\theoremstyle{definition}
\newtheorem{assumption}{Assumption}
\newtheorem{definition}{Definition}

\theoremstyle{plain}

\newtheorem{theorem}{Theorem}
\newtheorem{lemma}{Lemma}

\newcounter{remark}
\newenvironment{remark}{%
\par\vspace{3pt}\noindent\refstepcounter{remark}\textbf{Remark~\theremark:}}%
{\par\endtrivlist\unskip}

\newcounter{problem}
\newenvironment{problem}{%
\par\vspace{3pt}\noindent\refstepcounter{problem}\textbf{Problem~\theproblem:}}%

\usepackage[extramath,probability,reviewmark]{mylatexdefs}

\input{defs}

\usepackage{stfloats}%

\title{Stochastic Time-Optimal Trajectory Planning for Connected and Automated Vehicles in\\Mixed-Traffic Merging Scenarios}

\author{Viet-Anh Le, {\itshape{Student Member, IEEE}}, Behdad Chalaki, {\itshape{Member, IEEE}},\\Filippos N. Tzortzoglou, {\itshape{Student Member, IEEE}}, and Andreas A. Malikopoulos, {\itshape{Senior Member, IEEE}}
\thanks{This work was supported by the National Science Foundation (NSF) under Grants CNS-2149520 and CMMI-2219761.}
\thanks{Viet-Anh Le is with the Department of Mechanical Engineering, University of Delaware, Newark, DE 19716 USA and with the Systems Engineering, Cornell University, Ithaca, NY 14850 USA (e-mail: {\tt\small vl299@cornell.edu}).}
\thanks{Behdad Chalaki is with Honda Research Institute USA, Inc., Ann Arbor, MI 48103 USA (e-mail: {\tt\small behdad\_chalaki@honda-ri.com}). 
His contribution to this work draws from his research conducted during his PhD studies at the University of Delaware.
}
\thanks{Filippos N. Tzortzoglou and Andreas A. Malikopoulos are with the School of Civil and Environmental Engineering, Cornell University, Ithaca, NY 14850 USA (e-mail: {\tt\small ft253@cornell.edu}; {\tt\small amaliko@cornell.edu}).}
}

\begin{document}

\maketitle

\begin{abstract}

Addressing safe and efficient interaction between connected and automated vehicles (CAVs) and human-driven vehicles in a mixed-traffic environment has attracted considerable attention.
In this paper, we develop a framework for stochastic time-optimal trajectory planning for coordinating multiple CAVs in mixed-traffic merging scenarios.
We present a data-driven model, combining Newell's car-following model with Bayesian linear regression, for efficiently learning the driving behavior of human drivers online.
Using the prediction model and uncertainty quantification, a stochastic time-optimal control problem is formulated to find robust trajectories for CAVs.
We also integrate a replanning mechanism that determines when deriving new trajectories for CAVs is needed based on the accuracy of the Bayesian linear regression predictions.
Finally, we demonstrate the performance of our proposed framework using a realistic simulation environment. 
\end{abstract}
\begin{IEEEkeywords}
Connected and automated vehicles, mixed traffic, trajectory planning, stochastic control, Bayesian linear regression.
 \end{IEEEkeywords}

\section{Introduction}
\subsection{Motivation}

\IEEEPARstart{T}{he} advancements in connectivity and automation for vehicles present an intriguing opportunity to reduce energy consumption, greenhouse gas emissions, and travel delays while still ensuring safety requirements.
Numerous studies have demonstrated the advantages of coordinating connected and autonomous vehicles (CAVs) using control and optimization approaches across various traffic scenarios, \eg \cite{rios2018impact,zhang2019impact,ding2020penetration}.
In recent years, numerous control approaches have been presented for the coordination of CAVs, assuming a 100\% penetration rate. These approaches include time and energy-optimal control strategies 
\cite{Malikopoulos2020,malikopoulos2018decentralized,chalaki2020TITS,chalaki2020TCST,xiao2021decentralized,xu2022decentralized}, model predictive control \cite{hult2018optimal,kloock2019distributed,katriniok2022fully}, and reinforcement learning \cite{chalaki2020ICCA,Sumanth2021,zhang2023predictive} (see \cite{katriniok2022towards,rios2016survey,guanetti2018control,ersal2020connected} for surveys).
However, a transportation network with a 100\% CAV penetration rate is not expected to be realized by 2060 \cite{alessandrini2015automated}.
As CAVs will gradually and slowly penetrate the market and co-exist with human-driven vehicles (HDVs) in the following decades, addressing planning, control, and navigation for CAVs in mixed traffic, given various human driving styles, is imperative.

Several studies have shown that controlling individual automated vehicles (AVs) may not be sufficient to enhance the overall traffic condition.
For example, Wang \etal \cite{wang2022ego} showed that ego-efficient lane-changing control strategies for AVs (without coordination between vehicles) are beneficial to the entire traffic flow only if the penetration rate of AVs is less than 50\%.
Thus, AVs should be connected to share information and be coordinated to benefit the entire mixed traffic.
However, this problem imposes significant challenges for several reasons.
First, control methods for CAVs need to integrate human driving behavior and human-AV interaction to some extent. Approaches not accounting for these factors may result in conservative CAV behavior to prioritize safety, potentially leading to efficiency degradation.
Moreover, optimizing the behavior for CAVs requires not only some standard metrics such as safety, fuel economy, or average travel time but also social metrics like motion naturalness and human comfort \cite{bellem2016objective}, which can, at times, be challenging to quantify.
Finally, both learning and control methods must be computationally efficient and scalable for real-time implementation.
Therefore, in this paper, we aim to address the coordination problem for CAVs in mixed traffic while considering merging scenarios as a representative example.

\subsection{Literature Review}

In this section, we summarize the state of the art related to planning, control, and navigation for CAVs in mixed traffic.
A significant number of articles have considered connected cruise control or platoon formation for CAVs in mixed traffic,
\eg \cite{jin2018connected,jin2020analysis,ozkan2021socially,wang2022data,mahbub2021_platoonMixed,mahbub2023safety,fu2023cooperative,mahbub2022NHM}, where the main objective is to guarantee string stability between CAVs and HDVs. 
However, in this section, we focus more on research efforts that address the problem in traffic scenarios such as merging at roadways and roundabouts, crossing intersections, and lane-merging or passing maneuvers.
In these scenarios, vehicles must interact to complete the tasks not only safely but also efficiently, \eg improving travel time, avoiding gridlocks, and minimizing traffic disruption and human discomfort.
These problems present a more intricate challenge due to their multi-objective nature. 
The current state-of-the-art methods of planning, control, and navigation for CAVs in ``interaction-driven" mixed-traffic scenarios can be roughly classified into two main and emerging categories: \emph{reinforcement learning} and \emph{optimization-based methods}.

(1) \emph{Reinforcement learning (RL):} In approaches using RL, the aim is to learn control policies for CAVs, usually trained using deep neural networks and trajectories obtained from traffic simulation \cite{bouton2019cooperation,chen2023deep,zhou2022multi}. Although these approaches impose several technical challenges for implementation \cite{Malikopoulos2022a}, especially in situations where the problem has a nonclassical information structure \cite{Malikopoulos2021}, nevertheless, they can be used to provide a benchmark and draw helpful concluding remarks.
To enhance the social coordination of RL policies with human drivers, the concept of social value orientation (SVO) was incorporated into the reward functions \cite{toghi2021cooperative,toghi2022social,valiente2023learning,valiente2022robustness}. 
RL algorithms generally do not guarantee real-time safety constraints, so such approaches might need to be combined with other techniques for safety-critical control, such as with control barrier function \cite{udatha2023reinforcement}, shielding \cite{inala2021safe}, or lower-level model predictive control \cite{brito2022learning}.
Inverse RL \cite{kuderer2015learning} or imitation learning \cite{acerbo2020safe} have been used to learn the reward functions of human drivers and demonstrate how CAVs can perform human-like behaviors.

(2) \emph{Optimization-based methods:} Such methods include optimal control and model predictive control (MPC) to find the control actions for the CAVs.
A large number of control methods have been built upon MPC since it can handle multiple objectives and constraints, and it exploits the benefits of long-term planning and replanning at every time step for robustness against uncertainty caused by drivers.
A very common approach is game-theoretic MPC, where the behavior of human drivers is described through some objective functions. 
The problem formulation can be based on Stackelberg (or leader-follower) game \cite{fisac2019hierarchical,hang2020integrated,wang2021socially}, partially observable stochastic game \cite{sadigh2018planning}, or potential game \cite{liu2023potential,evens2022learning}, where the objective functions that best describe human driving behavior can be learned using inverse RL \cite{kuderer2015learning}.
Game-theoretic MPC can be integrated with social factors such as SVO, which were presented in \cite{schwarting2019social,zhao2021yield,li2023interaction}, to develop socially-compatible control designs for CAVs.
Another common MPC approach is stochastic MPC, where the driving behavior of HDVs is modeled as stochastic uncertainties, and then MPC problems are formulated as stochastic optimization problems \cite{hu2022active,nair2022stochastic_b,schuurmans2023safe}. 
Meanwhile, many recent studies have leveraged the advancement in deep learning-based human prediction models, which leads to learning-based MPC \cite{schmerling2018multimodal,bae2022lane,gupta2023interaction,venkatesh2023connected}.
MPC approach was combined with Hamilton-Jacobi reachability-based safety analysis in \cite{leung2020infusing,hu2022sharp,tian2022safety} to guarantee safety under worst-case control actions of the HDVs.
A mixed-integer MPC approach was considered in \cite{bhattacharyya2023automated,larsson2021pro} where binary variables are used to formulate the collision avoidance constraints. Another MPC approach with weight adaptation strategies for different human driving behaviors was reported in \cite{le2023acc,le2022cdc}.
Some recent studies used an optimal control framework based on Hamiltonian analysis for improving both time and energy efficiency simultaneously \cite{sabouni2023merging,li2023cooperative,chalaki2023minimally}. 

The aforementioned research efforts have addressed the planning, control, and navigation problems for CAVs at a single-vehicle level. At the same time, only a limited number of research articles have attempted to address the coordination problem in mixed traffic for multiple CAVs.
For example, Yan and Wu \cite{yan_reinforcement_2021} presented a multi-agent RL framework for CAVs in microscopic simulation while the human drivers are simulated by a car-following model.
Peng \etal \cite{peng_connected_2021} considered two CAVs that navigate multiple HDVs in a signal-free intersection and designed a deep RL policy to coordinate the CAVs.
Buckman \etal \cite{buckman_sharing_2019} presented a centralized algorithm for socially compliant navigation at an intersection, given the social preferences of the vehicles.
Liu \etal \cite{liu2023safety} presented a recursive optimal control method for mixed-traffic on-ramp merging utilizing a control barrier function and a control Lyapunov function.
Mixed-integer optimization has been considered in \cite{faris2022optimization,ghosh2022traffic} for deriving the optimal order of the vehicles crossing the intersections. 
In recent work \cite{le2023coordination}, we presented a control framework that aims to derive time-optimal trajectories for CAVs in a mixed-traffic merging scenario given the HDVs’ future trajectories predicted from Newell’s car-following model \cite{newell_simplified_2002}.
The time-optimal trajectories are then combined with a safety filter based on control barrier functions.

\begin{figure}[!t]
 \begin{center}
 \includegraphics[scale=0.56, bb = 150 380 600 680, clip=true]{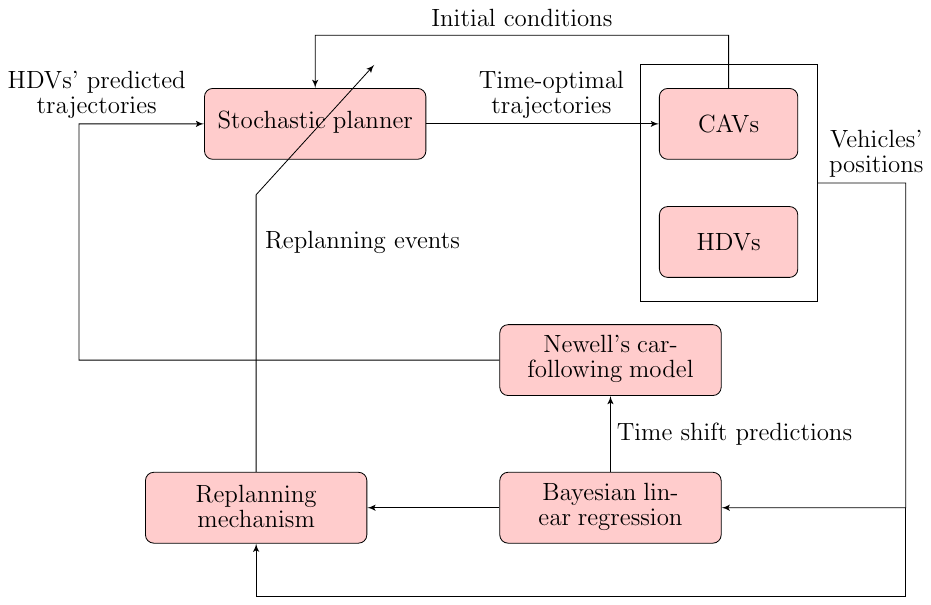}
 \end{center}
\caption{A diagram illustrating the proposed framework.} 
\label{fig:diagram}
\end{figure}

\subsection{Contributions and Organization}
In this paper, we propose a framework for trajectory planning based on stochastic control that can guarantee \emph{optimal, interaction-aware, robust, and safe} maneuvers for CAVs in mixed-traffic merging scenarios.
First, we consider the time-optimal control problem for trajectory planning of the CAVs, which utilizes the closed-form solution of a low-level energy-optimal control problem and satisfies state, input, and safety constraints \cite{Malikopoulos2020}.
Since the trajectory planning problem requires a specific prediction model for HDV's future trajectories, we use a \emph{data-driven Newell's car-following model} in which the time shift, a parameter that characterizes the personal driving behavior, is learned online using \emph{Bayesian linear regression} (BLR).
The use of data-driven Newell's car-following model with BLR allows us to not only predict the future trajectory and merging time of each HDV but also quantify the level of uncertainty in the predictions.
Using the predictions, we formulate a stochastic time-optimal control problem in which the safety constraints are formulated as probabilistic constraints for robustness without over-conservatism.
To address the potential discrepancy between the prediction model and the actual behavior of HDVs, we develop a \emph{replanning mechanism} based on checking the accuracy of the last stored BLR prediction for each HDV with the actual observation.
The overall structure of our proposed framework can also be illustrated in Fig.~\ref{fig:diagram}.
We validate the proposed framework's effectiveness in ensuring safe maneuvers and improving travel time through numerical simulations conducted in a commercial software.

In summary, the main contributions of this paper are threefold:
\begin{enumerate}[leftmargin=*]
\item We use a data-driven Newell's car-following model where BLR is utilized to calibrate the time shift for each human driver.
\item We formulate a stochastic time-optimal control problem with probabilistic constraints to derive robust trajectories for CAVs.
\item We develop a replanning mechanism based on assessing the accuracy of the BLR predictions.
\end{enumerate}

The remainder of the paper is organized as follows.
In Section~\ref{sec:problem}, we formulate the problem of coordinating CAVs in a mixed-traffic merging scenario and provide the preliminary on time-optimal trajectory planning in a deterministic setting. 
In Section~\ref{sec:human}, we present the data-driven Newell's car-following model with BLR for predicting the future trajectories of HDVs. 
In Section~\ref{sec:stochastic}, we develop a stochastic trajectory planning mechanism. Finally, in Section~\ref{sec:sim}, we numerically validate the effectiveness of the proposed framework in a simulation environment, and we draw concluding remarks in Section~\ref{sec:conc}.

\begin{figure*}[!t]
 \begin{center}
 \includegraphics[scale=0.48, bb = 0 50 950 510, clip=true]{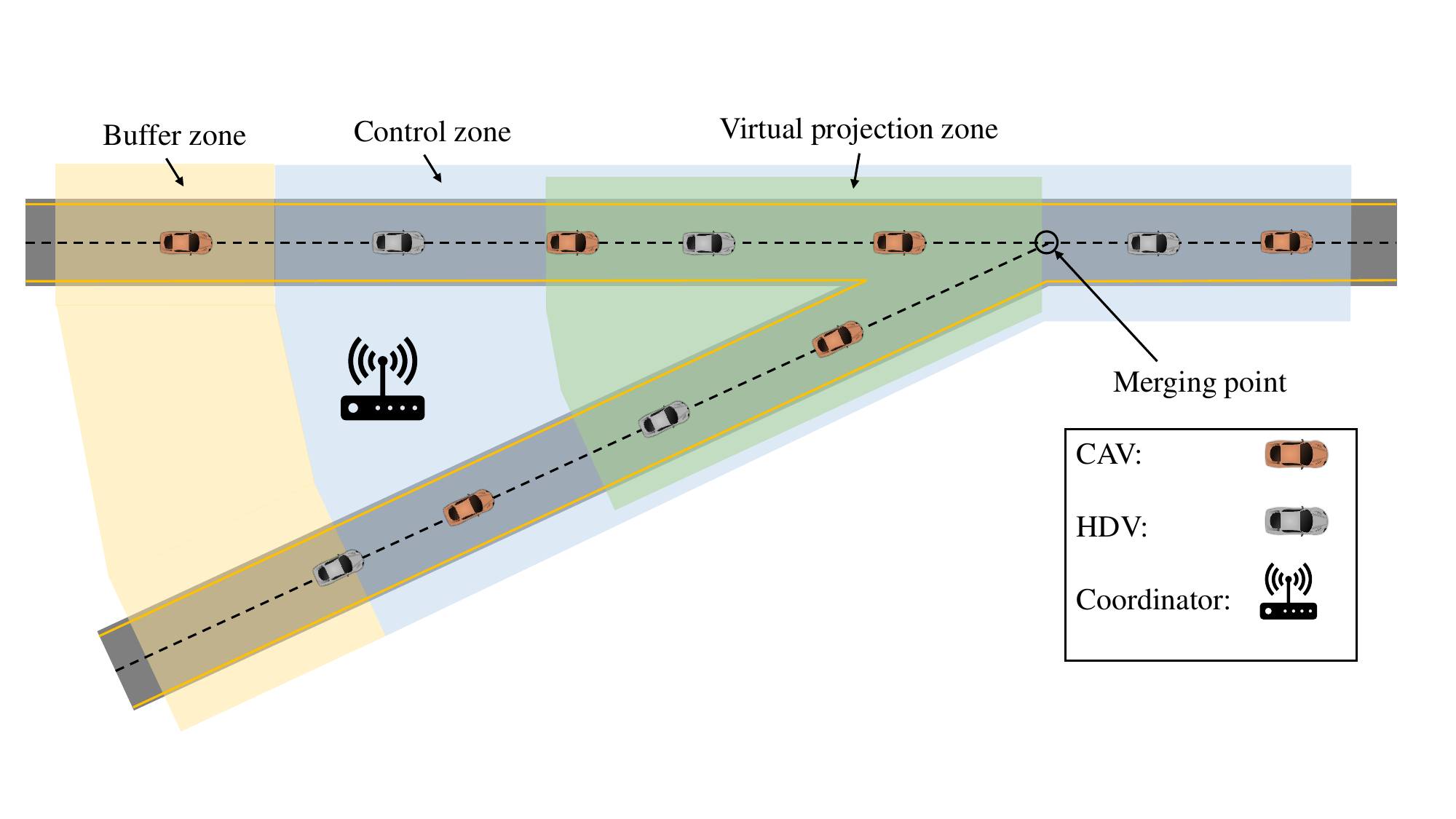}
 \end{center}
\caption{A merging scenario with two roadways intersecting at a merging point. The control zone and buffer zone are represented by the blue and yellow areas, respectively. In the virtual projection zone (green area), virtual projection is utilized (see Section~\ref{sec:human}).} 
\label{fig:scenario}
\end{figure*}

\section{Problem Formulation and Preliminaries}
\label{sec:problem}

In this section, we first introduce the problem of effectively coordinating multiple CAVs in a merging scenario, considering the presence of HDVs. 
Subsequently, we provide the preliminary materials on deterministic time-optimal trajectory planning for CAVs based on an earlier optimal control framework \cite{Malikopoulos2020}.

\subsection{Problem Formulation}

We consider the problem of coordinating multiple CAVs, co-existing with HDVs, in a merging scenario (Fig.~\ref{fig:scenario}), where two merging roadways intersect at a position called a merging point. 
We define a \emph{control zone} and a \emph{buffer zone}, located upstream of the control zone, which are represented by blue and yellow areas, respectively, in Fig.~\ref{fig:scenario}.
Within the control zone, the CAVs are controlled by the proposed framework, while in the buffer zone, CAVs are controlled using any adaptive cruise control methods.
We consider that a \textit{coordinator} is available who has access to the positions of all vehicles (including HDVs and CAVs). 
The coordinator starts collecting trajectory data of any HDVs once they enter the buffer zone so that at the control zone entry the data is sufficient for learning the first BLR model (see Section~\ref{sec:human}).
We consider that the CAVs and the coordinator can exchange information inside the control zone and buffer zone. 
Next, we provide some necessary definitions for our exposition.
\begin{definition}
\label{def:set_vehicle}
Let ${\LLL(t) = \{1,\ldots,L(t)\}}$, ${t\in\mathbb{R}_{\ge0}}$, be the set of vehicles traveling inside the control zone, where ${{L}(t)\in\mathbb{N}}$ is the total number of vehicles. 
Let ${\AAA(t) \subset \LLL(t)}$ and ${\HHH(t) \subset \LLL(t)}$ be the sets of CAVs and HDVs, respectively.
Note that the indices of the vehicles are determined by the order they enter the control zone. 
\end{definition}

\begin{definition}
For a vehicle $i \in \LLL(t)$, let ${\RRR_{i,\mathrm{S}}(t) \subset \LLL (t)}$ and ${\RRR_{i,\mathrm{N}}(t) \subset \LLL (t)}$, ${t\in\mathbb{R}_{\ge0}}$, be the sets of vehicles inside the control zone traveling on the same road as vehicle $i$ and on the neighboring road, respectively.
\end{definition}

Let $p^0$, $p^{\rm m}$, and $p^{\rm f} \in \RR$ be the positions of the control zone entry, the merging point, and the control zone exit, respectively.
Without loss of generality, we can set $p^{\rm m} = 0$.
We consider that the dynamics of each vehicle ${i \in \LLL (t)}$ are described by a double integrator model as follows
\begin{equation}\label{eq:model2}
\begin{split}
\dot{p}_{i}(t) &= v_{i}(t), 
\\
\dot{v}_{i}(t) &= u_{i}(t), 
\end{split} 
\end{equation}
where ${p_{i}\in\mathcal{P}}$, ${v_{i}\in\mathcal{V}}$, and
${u_{i}\in\mathcal{U}}$ denote the longitudinal position of the rear bumper, speed, and control input (acceleration/deceleration) of the vehicle, respectively. 
The sets ${\mathcal{P}, \mathcal{V},}$ and $\mathcal{U}$ are compact subsets of $\RR$. 
The control input is bounded by 
\begin{equation}\label{eq:uconstraint}
u_{\min} \leq u_{i}(t)\leq u_{\max},\quad \forall i \in \LLL(t),
\end{equation}
where ${u_{\min}<0}$ and ${u_{\max}>0}$ are the minimum and maximum control inputs, respectively, as designated by the physical acceleration and braking limits of the vehicles, or limits that can be imposed for driver/passenger comfort.
Next, we consider the speed limits of the CAVs, 
\begin{equation}\label{eq:vconstraint} 
v_{\min} \leq v_i(t) \leq v_{\max},\quad \forall i \in \LLL (t), 
\end{equation}
where ${v_{\min}> 0}$ and ${v_{\max}>0}$ are the minimum and maximum allowable speeds. 
Note that HDVs can violate the imposed speed limits. 
However, we make the following assumption.
\begin{assumption}
\label{assp:pos-speed}
The speed of HDVs is always positive, \ie ${v_j(t) > 0,\; \forall j \in \HHH (t)}$.
\end{assumption}

In practice, if HDVs come to a temporary full stop, Assumption~\ref{assp:pos-speed} can still be satisfied by assuming a sufficiently small lower bound on the speed.

Next, let $t_i^0$, $t_i^{\rm m}$, and $t_i^{\rm f} \in \mathbb{R}_{\ge0}$ be the times at which each vehicle $i$ enters the control zone, reaches the merging point, and exits the control zone, respectively.
To avoid conflicts between vehicles in the control zone, we impose two types of constraints: (1) lateral constraints between vehicles traveling on different roads; and (2) rear-end constraints between vehicles traveling on the same road.
Specifically, to prevent a potential conflict between \CAV{i} and a vehicle ${k \in \RRR_{i,\mathrm{N}}(t)}$ traveling on the neighboring road, we require a minimum time gap ${\delta_{\rm{l}} \in \mathbb{R}_{\ge0}}$ between the time instants ${t_i^{\rm f}}$ and $t_k^{\rm f}$ when the \CAV{i} and vehicle $k$ cross the merging point, \ie
\begin{equation}
\label{eq:lateral_constraint}
\abs{t_i^{\rm m} - t_k^{\rm m}} \geq \delta_{\rm{l}}. 
\end{equation}
To prevent rear-end collision between \CAV{i} and its immediate preceding vehicle $k$ traveling on the same road, \ie ${k = \max \, \{ j \in \RRR_{i,\mathrm{S}}(t) \; | \; j < i \}}$, we impose the following rear-end safety constraint:
\begin{equation}
\label{eq:rearend_constraint}
 p_{k}(t - \delta_{\rm r}) - p_{i}(t) \geq d_{\min},\, t \in [t_i^{0}, t_k^{\rm f}],
\end{equation}
where ${d_{\min} \in \mathbb{R}_{\ge0}}$ and ${\delta_{\rm r} \in \mathbb{R}_{\ge0}}$ are the minimum distance at a standstill and safe time gap. Note that $p_{k}(t - \delta_{\rm r})$ denotes the position of vehicle $k$ at time instant $t - \delta_{\rm r}$. 
In addition, we need to guarantee the rear-end safety constraint \eqref{eq:rearend_constraint} after the merging point between each \CAV{i} and a vehicle ${k \in \RRR_{i,\mathrm{N}} (t)}$ entering the control zone on the neighboring road and crosses the merging point immediately before \CAV{i} as follows
\begin{equation}
\label{eq:rearend_constraint-2}
p_{k}(t - \delta_{\rm r}) - p_{i}(t) \geq d_{\min},\, t \in [t_i^{\rm m}, t_k^{\rm f}],
\end{equation}
for ${k = \max \{ j \in \RRR_{i,\mathrm{N}}(t) \;|\; t_j^{\rm m} < t_i^{\rm m} \} }$.

\subsection{Time-Optimal Trajectory Planning}

Next, we explain the deterministic time-optimal trajectory planning framework developed initially for coordinating CAVs with a $100$\% penetration \cite{Malikopoulos2020}.
We start the exposition with the unconstrained solution of an energy-optimal control problem for each \CAV{i} \cite{malikopoulos2018decentralized}.
Given a fixed $t_i^{\rm f}$ that \CAV{i} exits the control zone, the energy-optimal control problem aims at finding the optimal control input (acceleration/deceleration) for each CAV by solving the following problem.
\begin{problem}\label{prb:energy-optimal-1}
(\textbf{Energy-optimal control problem})
Let $t_{i}^0$ and $t_{i}^{\rm f}$ be the times that \CAV{i} enters and exits the control zone, respectively. 
The energy-optimal control problem for \CAV{i} at $t_{i}^0$ is given by: 
\begin{equation}
\label{eq:energy_cost}
\begin{split}
&\minimize_{u_i(t)\in \UUU} \quad \frac{1}{2} \int_{t^{0}_{i}}^{t^{\rm f}_{i}} u^2_i(t)~\mathrm{d}t,
\\ 
& \subjectto 
\\
& \quad \eqref{eq:model2}, \eqref{eq:uconstraint},\eqref{eq:vconstraint},
\\
& \quad \eqref{eq:rearend_constraint},\; k = \max \, \{ j \in \RRR_{i,\mathrm{S}}(t_i^0) \; | \; j < i \}, 
\\
& \quad \eqref{eq:rearend_constraint-2},\; k = \max \, \{ j \in \RRR_{i,\mathrm{N}}(t_i^0) \; | \; t_j^{\rm m} < t_i^{\rm m} \},
\\
& \text{given:} 
\\
& \quad p_i (t_i^0) = p^0, \,\, v_i (t_i^0) = v_i^0, 
 \,\, p_i (t_i^{\rm f}) = p^{\rm f},
\end{split}
\end{equation}
where $v_i^0$ is the speed of \CAV{i} at the entry point.
The boundary conditions in \eqref{eq:energy_cost} are set at the entry and exit of the control zone. 
\end{problem}

The closed-form solution of Problem~\ref{prb:energy-optimal-1} for each \CAV{i} can be derived using the Hamiltonian analysis.
If none of the state and control constraints are active, the Hamiltonian becomes \cite{malikopoulos2018decentralized}
\begin{equation}\label{c1}
\begin{aligned}
H_i(t,p_i(t),v_i(t),u_i(t)) = \frac{1}{2}u_i^2(t) + \lambda_i^p v_i (t) + \lambda_i^v u_i(t),
\end{aligned}
\end{equation}
where \(\lambda_i^p\) and \(\lambda_i^v\) are co-states corresponding to position and speed, respectively.
The Euler-Lagrange equations of optimality are given by
\begin{align}
\dot{\lambda}_i^p & =-\frac{\partial H_i}{\partial p_i}=0, \label{euler:optimallambdav} 
\\ 
\dot{\lambda}_i^v &= -\frac{\partial H_i}{\partial v_i}= -\lambda_i^p, \label{euler:optimallambdap}
\\
\frac{\partial H_i}{\partial u_i} &= u_i+\lambda_i^v=0 \label{euler:optimalU}.
\end{align}

Using the Euler-Lagrange optimality conditions \eqref{euler:optimallambdav}-\eqref{euler:optimalU} to the Hamiltonian \eqref{c1}, we obtain the optimal control law and trajectory as follows
\begin{equation}\label{eq:optimalTrajectory}
\begin{split}
u_i(t) &= 6 \phi_{i,3} t + 2 \phi_{i,2}, \\
v_i(t) &= 3 \phi_{i,3} t^2 + 2 \phi_{i,2} t + \phi_{i,1}, 
\\
p_i(t) &= \phi_{i,3} t^3 + \phi_{i,2} t^2 + \phi_{i,1} t + \phi_{i,0},
\end{split}
\end{equation} 
where ${\phi_{i,3}, \phi_{i,2}, \phi_{i,1}, \phi_{i,0} \in \RR}$ are constants of integration.
Since the speed of \CAV{i} is not specified at the exit time $t_i^{\rm f}$, we consider the boundary condition \cite{bryson1975applied}
\begin{equation}\label{eq:lambdaV-sb}
 \lambda_i^v(t_i^{\rm f}) = 0.
\end{equation}
By substituting \eqref{eq:lambdaV-sb} into \eqref{euler:optimalU} at $t_i^{\rm f}$, \ie ${u_i(t_i^{\rm f}) + \lambda_i^v(t_i^{\rm f}) = 0}$, we obtain a terminal condition ${u_i(t_i^{\rm f}) = 0}$.
Given the boundary conditions in \eqref{eq:energy_cost} and ${u_i(t_i^{\rm f}) = 0}$, and considering $t_i^{\rm f}$ is known, the constants of integration can be found by:
\begin{equation}
\label{eq:phi_i}
\bbsym{\phi}_i =
\begin{bmatrix}
\phi_{i,3}
\\
\phi_{i,2}
\\
\phi_{i,1}
\\
\phi_{i,0}
\end{bmatrix}
= 
\begin{bmatrix}
(t_i^0)^3 & (t_i^0)^2 & t_i^0 & 1 
\\
3(t_i^0)^2 & 2t_i^0 & 1 & 0 
\\
(t_i^{\rm f})^3 & (t_i^{\rm f})^2 & t_i^{\rm f} & 1 
\\
6t_i^{\rm f} & 2 & 0 & 0 
\end{bmatrix}^{-1}
\begin{bmatrix}
p^0 
\\
v_i^0 
\\
p^{\rm f} 
\\
0
\end{bmatrix}.
\end{equation}

Note that using the Cardano's method \cite{cardano2007rules}, the time trajectory $t_i(p_i)$ as a function of the position is given by
\begin{multline}\label{eq:cardano}
 t_i(p_i) = \\
 \sqrt[3]{ - \frac{1}{2} \left(\omega_{i,1} + \omega_{i,2}~ p_i \right) + \sqrt{\frac{1}{4} \left(\omega_{i,1} + \omega_{i,2}~ p_i \right) ^ 2 + \frac{1}{27}\omega_{i,0} ^ 3}} +\\
\sqrt[3]{ - \frac{1}{2} \left(\omega_{i,1} + \omega_{i,2}~ p_i \right) - \sqrt{\frac{1}{4} \left(\omega_{i,1} + \omega_{i,2}~ p_i \right) ^ 2 + \frac{1}{27}\omega_{i,0} ^ 3}}\\
+ \omega_{i,3},\quad \quad p_{i} \in \mathcal{P},
 \end{multline}
 \begin{align}
 \omega_{i,0} &= \frac{\phi_{i,1}}{\phi_{i,3}} - \frac{1}{3}\left(\frac{\phi_{i,2}}{\phi_{i,3}}\right)^2, \label{eqn:omega0}\\
 \omega_{i,1} &= \frac{1}{27}\left[2\left(\frac{\phi_{i,2}}{\phi_{i,3}}\right)^3 - \frac{9 \phi_{i,2} \cdot \phi_{i,1}}{(\phi_{i,3}) ^ 2} \right] + \frac{\phi_{i,0}}{\phi_{i,3}} \label{eqn:omega1}, \\
 \omega_{i,2} &= - \frac{1}{\phi_{i,3}}, \quad\quad
 \omega_{i,3} = - \frac{\phi_{i,2}}{3 \phi_{i,3}},\label{eqn:omeg2,a3}
\end{align}
where $\omega_{i,3}, \omega_{i,2}, \omega_{i,1}$, and $\omega_{i,0}\in\mathbb{R}$ such that ${\frac{1}{4}(\omega_{i,1} + \omega_{i,2}~ p_i) ^ 2 + \frac{1}{27}\omega_{i,0} ^ 3 > 0}$, and they are all defined in terms of $\phi_{i,3}$, $ \phi_{i,2}$, $ \phi_{i,1}$, $ \phi_{i,0}\in\mathbb{R}$, with $\phi_{i,3} \neq 0$. 
The algebraic derivation of \eqref{eq:cardano} is tedious but standard \cite{Malikopoulos2020}, and thus omitted. 
We use \eqref{eq:cardano} to compute the merging time $t_i^{\rm m}$.

Next, we formulate the time-optimal control problem to minimize the travel time and guarantee all the constraints for CAVs given the energy-optimal trajectory \eqref{eq:optimalTrajectory} at $t_i^0$.
We enforce this unconstrained trajectory as a motion primitive to avoid the complexity of solving a constrained optimal control problem by piecing constrained and unconstrained arcs together \cite{malikopoulos2018decentralized}.
We refer to this problem as \emph{deterministic planning problem} to differentiate it from the stochastic problems discussed in the next sections.

\begin{problem} 
(\textbf{Deterministic planning at the control zone entry})
At the time $t_i^{0}$ of entering the control zone, let ${\mathcal{T}_i(t_i^0)=[\underline{t}_i^{\rm f}, \overline{t}_i^{\rm f}]}$ be the feasible range of travel time under the state and input constraints of \CAV{i} computed at $t_i^0$. 
The formulation for computing $\underline{t}_i^{\rm f}$ and $\overline{t}_i^{\rm f}$ can be found in \cite{chalaki2020experimental}.
Then \CAV{i} solves the following time-optimal control problem to find the minimum exit time $t_i^{\rm f} \in \mathcal{T}_i(t_i^0)$ that satisfies all state, input, and safety constraints
\label{prb:optimal_MZ}
\begin{align}
\begin{split}
 \label{eq:tif_1}
 &\minimize_{t_i^{\rm f} \in \mathcal{T}_i(t_i^0)} \quad t_i^{\rm f} 
 \\
 & \subjectto 
 \\
 & \quad \eqref{eq:uconstraint}, \eqref{eq:vconstraint}, 
 \\
 & \quad \eqref{eq:lateral_constraint},\; \forall \, k \in \RRR_{i,\mathrm{N}}(t_i^0), 
 \\
 & \quad \eqref{eq:rearend_constraint},\; k = \max \, \{ j \in \RRR_{i,\mathrm{S}}(t_i^0) \; | \; j < i \}, 
 \\
 & \quad \eqref{eq:rearend_constraint-2},\; k = \max \, \{ j \in \RRR_{i,\mathrm{N}}(t_i^0) \; | \; t_j^{\rm m} < t_i^{\rm m} \},
 \\ 
 & \quad \eqref{eq:optimalTrajectory}, 
 \\
 & \text{given:} 
 \\
 & \quad p_i (t_i^0) = p^0, \, v_i (t_i^0) = v_i^0, \\
 & \quad p_i (t_i^{\rm f}) = p^{\rm f}, \, u_i (t_i^{\rm f}) = 0.
\end{split}
\end{align}
\end{problem}

The computation steps for numerically solving Problem~\ref{prb:optimal_MZ} are summarized as follows or can be also found in \cite{chalaki2020experimental}.
First, we initialize ${t_i^{\rm f} = \underline{t}_i^{\rm f}}$, and compute the parameters $\bbsym{\phi}_i$ using \eqref{eq:phi_i}.
We evaluate all the state, control, and safety constraints. 
If none of the constraints is violated, we return the solution; otherwise, $t_i^{\rm f}$ is increased by a step size. 
The procedure is repeated until the solution satisfies all the constraints.
By solving Problem~\ref{prb:optimal_MZ}, the optimal exit time $t_i^{\rm f}$ along with the optimal trajectory and control law \eqref{eq:optimalTrajectory} are obtained for \CAV{i} for ${t\in [t_i^0, t_i^{\rm f}]}$.

\begin{remark}
If a feasible solution to Problem~\ref{prb:optimal_MZ} exists, then the solution is a cubic polynomial that guarantees none of the constraints become active.
In case the solution of Problem~\ref{prb:optimal_MZ} does not exist,
we can derive the optimal trajectory for the CAVs by piecing together the constrained and unconstrained arcs until the solution does not violate any constraints (see \cite{malikopoulos2018decentralized}).
\end{remark}

\section{Human Drivers’ Trajectory Prediction}
\label{sec:human}

To solve the trajectory planning problem for \CAV{i}, the trajectories and merging times of all vehicles having potential conflicts with \CAV{i} must be available.
When \CAV{i} enters the control zone, the time trajectories of all CAVs traveling inside the control zone can be obtained from the coordinator.
However, the time trajectories of the HDVs are not known.
Next, we propose an approach to predict the trajectories of the HDVs traveling inside the control zone by combining Newell's car-following model \cite{newell_simplified_2002} and BLR \cite[Chapter~3]{bishop2006pattern}.

\subsection{Bayesian Linear Regression}

Consider $N\in\mathbb{R}$ noisy observations $y_i$ of a linear model with Gaussian noise: $y_i = \bbsym{\theta}^\top \bbsym{x}_i + e_i$, for $i = 1, \dots, N$, where $\bbsym{\theta} \in \RR^M$, $M\in\mathbb{N},$ is the vector of weights, $\bbsym{x}_i \in \RR^M$ is the vectors of inputs, $e_i \sim \NNN (0, \beta^{-1})$ is the Gaussian noise where $\beta \in \mathbb{R}_{\ge0}$ is the precision (the inverse of variance).
Let $\OOO = (\bbsym{X}, \bbsym{Y})$ be the tuple of observation data for inputs and outputs, where $\bbsym{X} = [\bbsym{x}_1^\top, \dots, \bbsym{x}_N^\top]^\top \in \RR^{M\times N}$, $\bbsym{Y} = [y_1, \dots, y_N]^\top \in \RR^{N}$.
The goal of BLR is to find the maximum likelihood estimate for $\bbsym{\theta}$ given the observation data.

If we assume a Gaussian prior over the weights $\bbsym{\theta} \sim \NNN(\bbsym{0}, \alpha^{-1} \bbsym{I}_M)$ where $\alpha \in \mathbb{R}_{\ge0}$ is the precision and $\bbsym{I}_M$ is the $M \times M$ identity matrix, and the Gaussian likelihood $\probP (\bbsym{Y} \, |\, \bbsym{X}, \bbsym{\theta}, \beta) = \NNN (\bbsym{\theta}^\top \bbsym{X}, \beta^{-1} \bbsym{I}_N)$,
then this posterior distribution is
\begin{equation}
\probP (\bbsym{\theta} \,| \, \bbsym{Y}, \bbsym{X}, \alpha, \beta) \propto 
\probP (\bbsym{Y} \,|\, \bbsym{X}, \bbsym{\theta}, \beta)\, \probP (\bbsym{\theta} \,|\, \alpha).
\end{equation}
The likelihood function $\probP (\bbsym{Y} \, |\, \bbsym{X}, \bbsym{\theta})$ is computed by 
\begin{equation}
\probP (\bbsym{Y}\, | \, \bbsym{X}, \bbsym{\theta}, \beta) = \frac{1}{N} \prod_{i=1}^{N}\, \NNN (y_i \,|\, \bbsym{x}_i, \bbsym{\theta}, \beta).
\end{equation}
The log of the likelihood function can be written as
\begin{equation}
\log\, \probP (\bbsym{Y}\, | \, \bbsym{X}, \bbsym{\theta}, \beta) 
= \frac{N}{2} \log \beta - \frac{N}{2} \log 2\pi - \beta \EEE (\bbsym{\theta}),
\end{equation}
where $\EEE (\bbsym{\theta})$ is the sum-of-squares error function coming from the exponent of the likelihood function which is computed by
\begin{equation}
\EEE (\bbsym{\theta}) = \frac{1}{2} \sum_{i=1}^{N} (y_i - \bbsym{\theta}^\top \bbsym{x}_i)^2 = \frac{1}{2} \norm{\bbsym{Y} - \bbsym{X} \bbsym{\theta}}_2^2.
\end{equation}
Since the posterior is proportional to the product of likelihood and prior, the log of the posterior distribution is computed as follows
\begin{equation}
\label{eq-posterior}
\log\, \probP (\bbsym{\theta} \, | \, \bbsym{Y}, \bbsym{X}, \alpha, \beta) 
= - \beta \EEE (\bbsym{\theta}) - \frac{1}{2} \alpha \bbsym{\theta}^\top \bbsym{\theta} + c,
\end{equation}
where $c$ is a constant.
Therefore, the maximum-a-posteriori (MAP) estimate of $\bbsym{\theta}$ can be found by maximizing the log posterior \eqref{eq-posterior}, which has the following analytical solution: 
\begin{align}
\label{eq:dist-blr}
\bbsym{\mu}_{\bbsym{\theta}} &= \beta \bbsym{\Sigma}_{\bbsym{\theta}} \bbsym{X}^\top \bbsym{Y}, \\
\bbsym{\Sigma}_{\bbsym{\theta}}^{-1} &= \beta \bbsym{X}^\top \bbsym{X} + \alpha \bbsym{I}_m,
\end{align}
while estimates for priors, \ie $\alpha$ and $\beta$, can be obtained by the empirical Bayes method (also known as maximum marginal likelihood \cite{bishop2006pattern}).

Once the BLR model is trained, the posterior predictive distribution for $\bbsym{\theta}$ is a Gaussian distribution $\NNN(\bbsym{\mu}_{\bbsym{\theta}}, \bbsym{\Sigma}_{\bbsym{\theta}})$ with the mean and covariance matrix given in \eqref{eq:dist-blr}.
At a new input $\bbsym{x}_*$, the predicted mean and variance are given by
\begin{align}
\mu_{*} &= \bbsym{\mu}_{\bbsym{\theta}}^\top \bbsym{x}_*, \\
\sigma_*^2 &= \bbsym{x}_*^\top \bbsym{\Sigma}_{\bbsym{\theta}} \bbsym{x}_* + \beta^{-1}.
\end{align}

BLR is highly suitable for online learning implementation due to its light computation, where the complexity generally is $\OOO(M^2\,N)$, \ie it scales linearly with the training data size and quadratically with the input dimension. 
Moreover, we can check for retraining by comparing the new observation to the confidence interval or by considering prediction uncertainty. This approach can avoid overly frequent model retraining.

\subsection{Data-Driven Newell's Car-Following Model with Bayesian Linear Regression}

\begin{figure}[!t]
\begin{subfigure}{.245\textwidth}
\includegraphics[scale=0.38]{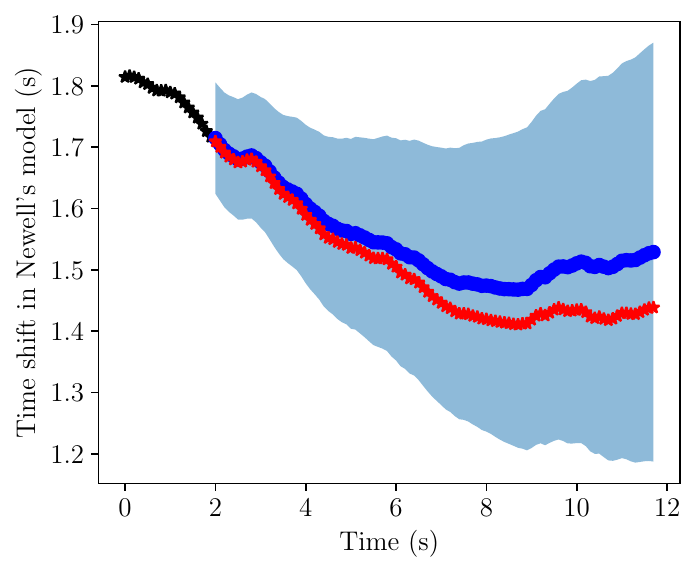}
\caption{}
\label{fig:blr_ex-a}
\end{subfigure}
\hspace{-8pt}
\begin{subfigure}{.245\textwidth}
\includegraphics[scale=0.38]{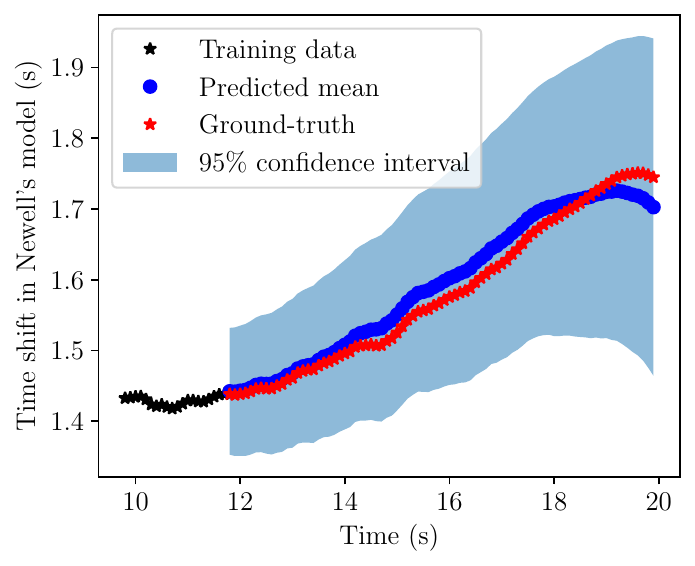}
\caption{}
\label{fig:blr_ex-b}
\end{subfigure}
\caption{The time shift prediction from BLR used to calibrate Newell's car-following model for a human-driven vehicle in Lyft level-5 open dataset \cite{li2023large}. 
The first model (a) is trained using the initial $20$ data points, and the second model (b) is trained with more recent data if the first model yields high predicted variances.} 
\label{fig:blr_ex}
\end{figure}

Newell's car-following model \cite{newell_simplified_2002} considers that the position of each vehicle is shifted in time and space from its preceding vehicle's trajectory due to the effect of traffic wave propagation.
Specifically, the position of each \HDV{k}, ${k \in \HHH (t)}$, is predicted from the position of its preceding vehicle $j$ as follows
\begin{equation}\label{eq:newell}
p_{k}(t) = p_{j} \big( t - \tau_{k} \big) - w\, \tau_{k},
\end{equation}
where ${\tau_{k} \in \mathbb{R}_{\ge0}}$ is the time shift of \HDV{k}, and ${w \in \mathbb{R}_{\ge0}}$ is the speed of the backward propagating congestion waves, which is considered to be a constant \cite{wong2021traffic,molnar_board_2022}.
The time shift ${\tau_{k}}$ is considered as a stochastic variable and can be learned by BLR.
Since ${v_{j} (t) > 0}$ (Assumption~\ref{assp:pos-speed}) and ${w > 0}$, ${p_{j} \left( t - \tau_{k} \right) - w \tau_{k}}$ is a strictly decreasing function of $\tau_{k}$. 
Thus, there exists a unique value of $\tau_{k}$ such that \eqref{eq:newell} is satisfied for any $t$.
In this paper, rather than only using the data at a specific time instant, we use the observations over a finite estimation horizon of length $H \in \ZZplus$ to estimate the distribution of $\tau_k$ for each \HDV{k} by a BLR model as follows
\begin{equation}
\label{eq:blrmodel}
\tau_k \sim \BLR_k (\bbsym{x}_k; \bbsym{\theta}_k),
\end{equation}
where $\BLR_k$ denotes the BLR model, $\bbsym{x}_k = [1, p_k, p_j]^\top \in \RR^3$ is the vector of inputs, and $\bbsym{\theta}_k \in \RR^3$ is the vector of weights.
Henceforth, for ease of notation, we use $\BBB_k (p_k, p_j)$ to denote the BLR model for $\tau_k$ given a preceding vehicle $j$.

To demonstrate the model's capability to accurately learn realistic human driving behavior, we utilized the trajectory data for a specific human driver in Lyft level-5 open dataset \cite{li2023large} whose actual time shift varies between \SI{1.4}{s} and \SI{1.8}{s}.
The predicted time shift with $95\%$ confidence interval using BLR is shown in Fig.~\ref{fig:blr_ex}.
We utilized $N = 20$ initial data points to train a BLR model (Fig.~\ref{fig:blr_ex-a}) and retrain the model (Fig.~\ref{fig:blr_ex-b}) with more recent data if either the BLR prediction uncertainty is too high or the actual observations are outside the $95\%$ confidence interval.

Note that to capture the lateral interaction of each HDV with vehicles on the neighboring road in the merging scenario, we consider the virtual projection of vehicles traveling on that road.
The virtual projection is implemented in a proximity area before the merging point, defined as the virtual projection zone in Fig.~\ref{fig:scenario}.
The virtual projection is illustrated by an example shown in Fig.~\ref{fig:projection}. 
We consider that from the perspective of \HDV{3}, the projected \CAV{2} is the preceding vehicle instead of \CAV{1}.
Similar generalized car-following models for capturing the merging behavior of human drivers have been presented in \cite{holley2023mr,kreutz2021analysis,guo2019improved}.

\begin{figure}[!t]
 \begin{center}
 \includegraphics[scale=0.48, bb = 250 170 750 455, clip=true]{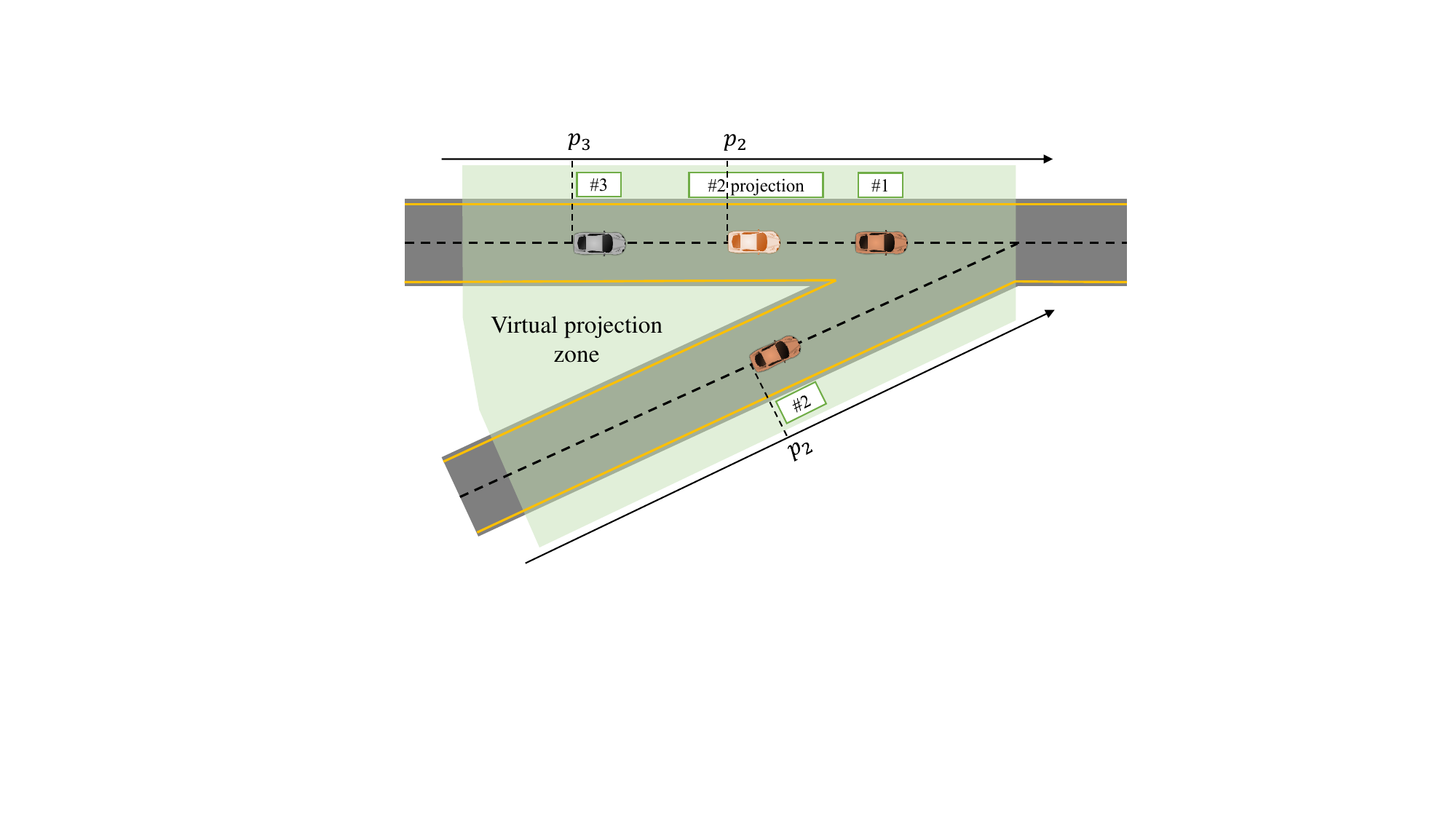}
 \end{center}
\caption{An example of virtual projection in the virtual projection zone, where the \CAV{2} is projected from the perspective of \HDV{3}. } 
\label{fig:projection}
\end{figure} 

\subsection{Exception Handling}

Next, we present a method to handle the case when an HDV, \eg \HDV{k}, is not preceded by any vehicles in the control zone, including those determined by virtual projection.
Generally, it is reasonable to assume that \HDV{k} remains its current speed in this case. 
However, to further quantify the uncertainty in human driving behavior by exploiting the data-driven Newell's car-following model, we consider that \HDV{k} follows \emph{a virtual preceding vehicle} with a constant speed trajectory. %
Let $k'$ denote the virtual preceding vehicle to \HDV{k}. 
The constant speed trajectory of the virtual preceding vehicle is given by
\begin{align}
p_{k'} (t) &= \phi_{k',1}\, t + \phi_{k',0}, \\
v_{k'} (t) &= \phi_{k',1}, 
\end{align}
where $\phi_{k',1} \in \mathbb{R}_{\ge0}$ and $\phi_{k',0} \in \RR$ are constants, in which $\bar{\phi}_{k',1}$ is computed based on the average speed of \HDV{k} over the estimation horizon, while $\phi_{k',0}$ is chosen such that $p_{k'} (t_k^0 - \bar{\tau}) = p^0$ with $\bar{\tau} \in \mathbb{R}_{\ge0}$ is an arbitrarily predefined constant.
We consider that the actual position trajectory of \HDV{k} is computed by Newell's car-following model given the virtual preceding vehicle $k'$ as follows
\begin{equation}
\label{eq:const-speed-stochastic}
p_k (t) = p_{k'} (t - \tau_k) - w \tau_k = \phi_{k',1}\, (t - \tau_k) + \phi_{k',0} - w \tau_k, 
\end{equation}
where we quantify $\tau_k$ with a BLR model $\tau_k \sim \BLR (p_k, p_{k'})$, which is similar to \eqref{eq:blrmodel}. 

\section{Stochastic Planning with probabilistic constraints}
\label{sec:stochastic}

In this section, we develop a stochastic trajectory planning framework using the data-driven Newell's following model for learning human driving behavior presented in the last section.
The use of stochastic control can reduce the conservatism of classical robust control for uncertain systems by formulating robust constraints as probabilistic constraints \cite{brudigam2021stochastic}.
As a result, probabilistic constraints have been used recently in robust trajectory optimization algorithms, \eg \cite{lew2020chance,coulson2021distributionally,hewing2019cautious}. 

\subsection{Uncertainty Quantification}
\label{sub:uncertainty}

\begin{remark}
In our framework, we consider that the trajectories of CAVs are deterministic, or equivalently, stochastic variables with zero variances. 
\end{remark}

Note that given the data-driven Newell's car-following model using BLR, the time shift $\tau_k$ of \HDV{k} at any future time $t$ must satisfy the following equation
\begin{equation}
\label{eq:tau-dist}
\tau_k (t) \sim \BLR \Big( p_j \big( t - \tau_k (t) \big) - w \tau_k (t),\, p_j(t) \Big),
\end{equation}
where $j$ is the index of the preceding vehicle.
Solving \eqref{eq:tau-dist} to obtain a closed-form solution for $\tau_k (t)$ and $p_k (t)$ at any future time $t$ is computationally intractable. 
As a result, in what follows, we propose a method to simplify the predictions of trajectory and merging time for each \HDV{k} along with quantifying the uncertainty of the predictions.

When \HDV{k} enters the control zone at $t_k^0$, we train the first BLR model $\BLR_k$ for $\tau_k$ using a dataset of $H\in\mathbb{N}$ data points collected in the buffer zone.
Let $\tau_{k} (t_k^0) \sim \NNN (\mu_{\tau_{k}}, \sigma^2_{\tau_{k}})$ be the prediction of $\tau_{k}$ with the mean $\mu_{\tau_{k}}$ and the variance $\sigma^2_{\tau_{k}}$.
We utilize $\tau_{k} (t_k^0)$ to construct a nominal predicted trajectory and merging time for \HDV{k}.
In our analysis, we consider the \emph{zero-variance method} for approximating uncertainty propagation while making BLR prediction, which implies that if the inputs of a BLR model include a stochastic variable, we only use its mean to compute the mean and variance of the model output without taking its variance into account. The zero-variance method has been considered in multiple studies on using stochastic processes in control, \eg \cite{nghiem2017data,boedecker2014approximate}.

\begin{assumption}
\label{assp:UP}
The effect of uncertainty propagation is approximated by the zero-variance method.
\end{assumption}
Assumption~\ref{assp:UP} implies that the trajectory prediction of any HDV only depends on the uncertainty resulting from the time shift prediction of Newell's car-following model, and does not depend on the uncertainty in trajectory prediction of its preceding vehicles.
The reason for ignoring full uncertainty propagation is that it may lead to overly conservative constraints if the CAV penetration rate is low.
Moreover, Assumption~\ref{assp:UP} aims to simplify the computation since the distribution of vehicles' trajectories, which are inputs of the BLR model, is generally not Gaussian as we show later (\cf Lemma~\ref{lem:1}). 

Next, we show that the predicted position mean for any HDVs using the data-driven Newell's car-following model is either a cubic polynomial or an affine polynomial.

\begin{lemma}
\label{lem:1}
Given Assumption~\ref{assp:UP}, at any time $t$, if the distribution $\NNN(\mu_{\tau_k} (t) , \sigma^2_{\tau_k} (t))$ for $\tau_k (t)$ is known, 
and \HDV{k} is preceded by a vehicle $j$ whose predicted position mean is a cubic polynomial of time parameterized by ${\bbsym{\phi}_j = [\phi_{j,3}, \phi_{j,2}, \phi_{j,1}, \phi_{j,0}]^\top}$,
then the predicted position mean and variance of \HDV{k} at time $t$ can be computed as given by \eqref{eq:mu-f} and \eqref{eq:sigma-f}, where $\lambda_k = t- \mu_{\tau_k}$.
\begin{figure*}
\begin{align}
& \mu_{p_k}(t) = \phi_{j,3} (\lambda_k^3 + 3\lambda_k \sigma_{\tau_k}^2) + \phi_{j,2} (\lambda_k^2 + \sigma_{\tau_k}^2) + (\phi_{j,1}+w) \lambda_k + (\phi_{j,0} - wt)
\label{eq:mu-f} \\
&
\begin{multlined}
\sigma^2_{p_k}(t) = \sigma_{\tau_k}^2 
\Big( (\phi_{j,1}+w)^2 
+ 4 (\phi_{j,1}+w) \phi_{j,2} \lambda_k 
+ 6 (\phi_{j,1}+w) \phi_{j,3} \lambda_k^2 
+ 6 (\phi_{j,1}+w) \phi_{j,3} \sigma_{\tau_k}^2 
+ 4 \phi_{j,2}^2 \lambda_k^2 
+ 2 \phi_{j,2}^2 \sigma_{\tau_k}^2 \\
+ 12 \phi_{j,2} \phi_{j,3} \lambda_k^3 
+ 24 \phi_{j,2} \phi_{j,3} \lambda_k \sigma_{\tau_k}^2 
+ 36 \phi_{j,3}^2 \lambda_k^2 \sigma_{\tau_k}^2 
+ 15 \phi_{j,3}^2 \sigma_{\tau_k}^4
+ 9 \phi_{j,3}^2 \lambda_k^4
\Big).
\end{multlined} 
\label{eq:sigma-f}
\end{align}
\hrulefill
\end{figure*}
Moreover, the coefficients of the polynomial are computed as follows
\begin{equation}
\label{eq:hdv-poly}
\begin{split}
&\phi_{k,3} = \phi_{j,3}, \\
&\phi_{k,2} = \phi_{j,2} - 3\, \phi_{j,3}\, \mu_{\tau_k}, \\
&\phi_{k,1} = \phi_{j,1} - 2\, \phi_{j,2}\, \mu_{\tau_k} + 3\, \phi_{j,3}\, (\mu_{\tau_k}^2 + \sigma_{\tau_k}^2), \\
&\begin{multlined}
\phi_{k,0} = \phi_{j,0} - (\phi_{j,1} + w) \mu_{\tau_k} + \phi_{j,2} (\mu_{\tau_k}^2 + \sigma_{\tau_k}^2) \\
- \phi_{j,3} \mu_{\tau_k} (\mu_{\tau_k}^2 + 3\sigma_{\tau_k}^2).
\end{multlined}
\end{split} 
\end{equation}

\end{lemma}

\begin{proof}

The proof is given in Appendix~\ref{sec:lem1}.
Note that the distribution of $p_k(t)$ in this case is not Gaussian. 
\end{proof}

\begin{lemma}
\label{lem:2}
Given Assumption~\ref{assp:UP}, if the distribution $\NNN(\mu_{\tau_k} , \sigma^2_{\tau_k})$ for $\tau_k (t)$ at time $t$ is known, 
and \HDV{k} is preceded by a vehicle $j$ whose predicted position mean is an affine polynomial of time parameterized by ${\bbsym{\phi}_j = [\phi_{j,1}, \phi_{j,0}]^\top}$, i.e., it cruises with constant speed, 
then the predicted position mean and variance of \HDV{k} at time $t$ can be computed as given by
\begin{align}
\mu_{p_k} (t) &= (\phi_{j,1}+w)\, \lambda_k + (\phi_{j,0} - w\, t), \\
\sigma^2_{p_k} (t) &= (\phi_{j,1}+w)^2 \sigma^2_{\tau_k},
\end{align}
where $\lambda_k = t - \mu_{\tau_k}$ and the coefficients of the polynomial are computed as follows
\begin{equation}
\label{eq:hdv-poly-1}
\begin{split}
\phi_{k,1} &= \phi_{j,1}, \\
\phi_{k,0} &= \phi_{j,0} - (\phi_{j,1} + w) \mu_{\tau_k}.
\end{split}
\end{equation}
\end{lemma}

\begin{proof}
This is a trivial case of Lemma~\ref{lem:1} with ${\phi_{j,3} = \phi_{j,2} = 0}$.
\end{proof}

\begin{theorem}
The mean prediction for the position of any \HDV{k} is either a cubic polynomial or an affine polynomial of time.
\end{theorem}
\begin{proof}
Given Lemmas~\ref{lem:1} and \ref{lem:2}, if \HDV{k} is preceded by an HDV, \eg \HDV{j}, and the mean prediction for the position of \HDV{j} is either a cubic polynomial or an affine polynomial of time, then that of \HDV{k} is also either a cubic polynomial or an affine polynomial of time.
Therefore, we only need to consider the cases where (1) \HDV{k} is preceded by a CAV, \eg \CAV{i}, or (2) \HDV{k} is not preceded by any vehicle inside the control zone.
\begin{itemize}[leftmargin=*]
\item \textit{Case 1:} If \HDV{k} is preceded by \CAV{i}, since the position trajectory of \CAV{i} is a cubic polynomial, from Lemma~\ref{lem:1} we can verify that the predicted position mean of \HDV{k} is a cubic polynomial of time.
\item \textit{Case 2:} If \HDV{k} is not preceded by any vehicle inside the control zone and has not crossed the merging point, from \eqref{eq:const-speed-stochastic}, the predicted mean of ${p_k} (t)$ is
\begin{equation}
\mu_{p_k} (t) = \bar{\phi}_{j,1}\, t + (\phi_{j,0} - \bar{\phi}_{j,1}\, \mu_{\tau_k} - w \, \mu_{\tau_k}),
\end{equation}
\end{itemize}

which is a linear function of time.
\end{proof}

\begin{lemma}
\label{lem:merging_time}
Suppose \HDV{k} has not crossed the merging point. Then, the merging time of \HDV{k} is computed by
\begin{equation}
\label{eq:t_f_blr}
t_k^{\rm{m}} = t_j (w \tau_k) + \tau_{k} ,
\end{equation}
where $t_j (w \tau_k)$ denotes the time that the preceding vehicle $j$ reaches the position $w \tau_k$, where $\tau_k = \tau_k (t_k^0) \sim \NNN(\mu_{\tau_k}, \sigma^2_{\tau_k})$.
\end{lemma}
\begin{proof}
Evaluating Newell's car following model \eqref{eq:newell} at $t_k^{\rm m}$ we have 
\begin{equation}
 p_k (t_k^{\rm m}) = p_j(t_k^{\rm m}-\tau_k)-w \tau_k. 
\end{equation}
At the merging time, we have $p_k(t_k^{\rm m}) = p^{\rm m} = 0$ which results in 
\begin{equation}
 p_j(t_k^m-\tau_k)=w \tau_k. 
\end{equation}
As the speed of vehicle $k$ is always positive given \eqref{eq:vconstraint} and Assumption~\ref{assp:pos-speed}, we can compute the inverse as 
 \begin{equation}
 t_k^m-\tau_k=t_j(w \tau_k), 
\end{equation}
and the proof is complete. 
\end{proof}

If vehicle $j$ is an HDV, we approximate $t_j (w \tau_k)$ by solving $\mu_{p_j} = w \tau_k$.
If $\mu_{p_j} (t)$ is an affine polynomial of time parameterized by ${\bbsym{\phi}_j = [\phi_{j,1}, \phi_{j,0}]^\top}$, the time trajectory as a function of position is given by
\begin{equation}
t_j (p_j) = \frac{-\phi_{j,0}}{\phi_{j,1}}\, p_j,
\end{equation}
while if $\mu_{p_j} (t)$ is a cubic polynomial, the time trajectory follows Cardano formulation \eqref{eq:cardano}.
Given Lemma~\ref{lem:merging_time}, $t_k^{\rm{m}}$ is a stochastic variable with Gaussian distribution, ${t_k^{\rm{m}} \sim \NNN(\mu_{t_k^{\rm m}}, \sigma^2_{t_k^{\rm m}})}$ with $\mu_{t_k^{\rm m}} = \mu_{\tau_k} + t_j (w \tau_k)$ and ${\sigma^2_{t_k^{\rm m}} = \sigma^2_{\tau_k}}$. 
To guarantee that the computation of $t_j (w \tau_k)$ using the polynomial trajectories is valid, the position $w \tau_k$ must be inside the control zone. Thus, we impose the following assumption.
\begin{assumption}
\label{assp:wave}
The speed of the backward propagating congestion waves $w$ is chosen such that $w\, \tau_k \le p^{\rm f}$. 
\end{assumption}

Assumption~\ref{assp:wave} can be satisfied in practice since the term $w\, \tau_k$ describes the standstill spacing between vehicles and should be relatively small compared to the length from the merging point to the control zone exit.

\subsection{Stochastic Time-Optimal Control Problem with probabilistic constraints}

Since the predicted trajectory and merging time for any \HDV{k} are stochastic variables, next we formulate probabilistic constraints for rear-end and lateral safety that guarantee constraint satisfaction at a certain probability.
Let $\xi \in (0, 1)$ be the probability of constraint satisfaction.
The lateral probabilistic constraint for \CAV{i} and \HDV{k} entering from different roads is given by
\begin{equation}
\label{eq:cc-lateral}
\begin{split}
&\probP \big[t_i^{\rm{f}} - t_k^{\rm{f}} \ge \delta_l \big] \ge \xi, \\
\text{OR}\quad &\probP \big[t_i^{\rm{f}} - t_k^{\rm{f}} \le -\delta_l \big] \ge \xi.
\end{split}
\end{equation}
The deterministic rear-end constraints for \CAV{i} and its immediate preceding \HDV{k} in \eqref{eq:rearend_constraint} and \eqref{eq:rearend_constraint-2} are considered as the following probabilistic constraints
\begin{equation}
\label{eq:cc-rearend-1}
\probP \big[ p_i(t) - p_k(t - \delta_r) \le - d_{\min} \big] \ge \xi, \; \forall t \in [t_i^{0}, t_k^{\rm f}] 
\end{equation}
for ${k = \max \, \{ j \in \RRR_{i,\mathrm{S}}(t) \; | \; j < i \}}$, and
\begin{equation}
\label{eq:cc-rearend-2}
\probP \big[ p_i(t) - p_k(t - \delta_r) \le - d_{\min} \big] \ge \xi, \; \forall t \in [t_i^{\rm m}, t_k^{\rm f}],
\end{equation}
for ${k = \max \, \{ j \in \LLL(t) \; | \; t_j^{\rm m} < t_i^{\rm m} \}}$.

Therefore, we formulate the following stochastic time-optimal control problem for planning at the control zone entry.

\begin{problem} 
(\textbf{Stochastic planning at the control zone entry})
At the time $t_i^{0}$ of entering the control zone, \CAV{i} solves the following time-optimal control problem
\label{prb:cc-optimal_MZ}
\begin{align}
\label{eq:cc-optimal_MZ}
\begin{split}
 &\minimize_{t_i^{\rm f} \in \mathcal{T}_i(t_i^0)} \quad t_i^{\rm f} 
 \\
 & \subjectto 
 \\
 & \quad \eqref{eq:uconstraint}, \eqref{eq:vconstraint}, \eqref{eq:optimalTrajectory},
 \\
 & \quad \eqref{eq:cc-lateral},\; \forall \, k \in \RRR_{i,\mathrm{N}}(t_i^0), 
 \\
 & \quad \eqref{eq:cc-rearend-1},\; k = \max \, \{ j \in \RRR_{i,\mathrm{S}}(t_i^0) \; | \; j < i \}, 
 \\
 & \quad \eqref{eq:cc-rearend-2},\; k = \max \, \{ j \in \RRR_{i,\mathrm{N}}(t_i^0) \; | \; t_j^{\rm m} < t_i^{\rm m} \},
 \\
 & \text{given:} 
 \\
 & \quad p_i (t_i^0) = p^0, \,\, v_i (t_i^0) = v_i^0, \\
 & \quad p_i (t_i^{\rm f}) = p^{\rm f}, \,\, u_i (t_i^{\rm f}) = 0. 
\end{split}
\end{align}
\end{problem}

Given the uncertainty quantification of stochastic variables we derived in Section~\ref{sub:uncertainty} and the constraint tightening technique \cite{hewing2019cautious}, the lateral probabilistic constraint \eqref{eq:cc-lateral} is equivalent to the following deterministic form 
\begin{equation}
\label{eq:cc-lateral-a}
\begin{split}
& t_i^{\rm{f}} - \mu_{t_k^{\rm{f}}} \ge \delta_l + z \sigma_{t_k^{\rm{f}}}, \\
\text{OR}\quad & t_i^{\rm{f}} - \mu_{t_k^{\rm{f}}} \le -\delta_l - z \sigma_{t_k^{\rm{f}}},
\end{split}
\end{equation}
where $z = \sqrt{2}\, \mathrm{erf}^{-1}(2\xi-1)$ with $\mathrm{erf}^{-1}(\cdot)$ is the inverse error function. 
Likewise, the rear-end probabilistic constraints \eqref{eq:cc-rearend-1} and \eqref{eq:cc-rearend-2} can be respectively transformed to deterministic constraints as follows 
\begin{equation}
\label{eq:cc-rearend-1a}
p_i(t) - \mu_{p_k} (t - \delta_r) \le - d_{\min} - z \sigma_{p_k} (t - \delta_r), \, \forall t \in [t_i^0, t_k^{\rm f}],
\end{equation}
and 
\begin{equation}
\label{eq:cc-rearend-2a}
p_i(t) - \mu_{p_k} (t - \delta_r) \le - d_{\min} - z \sigma_{p_k} (t - \delta_r), \, \forall t \in [t_i^{\rm m}, t_k^{\rm f}].
\end{equation}

Thus, for solving Problem~\ref{prb:cc-optimal_MZ}, the probabilistic constraints in \eqref{eq:cc-optimal_MZ} are replaced by \eqref{eq:cc-lateral-a}, \eqref{eq:cc-rearend-1a}, and \eqref{eq:cc-rearend-2a}, which result in an equivalent deterministic optimization problem.

\subsection{Replanning}
\label{sub:replanning}

\begin{algorithm}[tb!]
\small
\caption{Replanning mechanism at time $t^{\rm c}$}
\label{alg:replan}
\begin{algorithmic}[1]
\Require $t^{\rm c}$, $\zeta$, $\tilde{t}_k, \,{\forall k \in \HHH(t^{\rm c})}$

\State $\texttt{Replan} \leftarrow \texttt{False}$
\For {$k \in \HHH (t^{\rm c})$}
	\State Compute $\hat{\tau}_k (t^{\rm c})$
 \If {$t^{\rm c} = t_k^0$}
	\State $\tilde{t}_k \leftarrow t^{\rm c}$
 \State Train $\BLR_k$
	\State Compute and store $\mu_{\tau_k} (\tilde{t}_k), \sigma^2_{\tau_k} (\tilde{t}_k) $ %
 	\State Predict $t_k^{\rm m}$ and $\bbsym{\phi}_k$ using \eqref{eq:t_f_blr} and \eqref{eq:hdv-poly}/\eqref{eq:hdv-poly-1}
	\ElsIf {$\hat{\tau}_k (t^{\rm c}) \notin \mathrm{CI}_\zeta (\tau_k (\tilde{t}_k))$ or $\texttt{Replan} = \texttt{True}$}
	\State $\tilde{t}_k \leftarrow t^{\rm c}$
 \State Retrain $\BLR_k$
	\State Compute and store $\mu_{\tau_k} (\tilde{t}_k), \sigma^2_{\tau_k} (\tilde{t}_k) $ %
	\State Predict $t_k^{\rm m}$ and $\bbsym{\phi}_k$ using \eqref{eq:t_f_blr} and \eqref{eq:hdv-poly}/\eqref{eq:hdv-poly-1}
	\State $\texttt{Replan} \leftarrow \texttt{True}$
	\EndIf
\EndFor
\State Construct $\AAA' (t^{\rm c})$ using Definition~\ref{def:replan-set}
\For {$i \in \AAA(t^{\rm c})$}
	\If {$t^{\rm c} = t_i^0$}
	\State Solve Problem~\ref{prb:cc-optimal_MZ} given $p_i(t_i^0)$ and $v_i(t_i^0)$
	\ElsIf {\texttt{Replan} and $i \in \AAA' (t^{\rm c})$}
	\State Solve Problem~\ref{prb:cc-replan} given $p_i(t^{\rm c})$ and $v_i(t^{\rm c})$
	\EndIf 
\EndFor
\end{algorithmic}
\end{algorithm} 

Since the future trajectory and merging time for any HDV derived in Section~\ref{sub:uncertainty} are computed based on the prediction of $\tau_k$ at $t_k^0$, the predictions are not reliable if $\hat{\tau}_k (t)$ at $t > t_k^0 $, where $\hat{\tau}_k$ denotes the actual observation of $\tau_k$ obtained by solving \eqref{eq:newell}, is highly different to $\tau_k (t_k^0)$. 
Under this discrepancy, the planned trajectories for the CAVs may not always ensure safe maneuvers.
To address this issue, next, we present a mechanism for replanning based on checking the accuracy of the BLR predictions.
First, we define \emph{replanning instances} and how to determine replanning instances as follows.
\begin{definition}
\label{def:replanning}
A time instance $t^{\rm{c}} \in\mathbb{R}_{\ge0}$ is a replanning instance if at $t^{\rm{c}}$ we need to replan for the CAVs in the control zone. 
At any time $t^{\rm{c}}$, we check whether $t^{\rm{c}}$ is a replanning instance if there exits \HDV{k} $\in \HHH(t^{\rm{c}})$ such that $\hat{\tau}_k (t^{\rm{c}}) \notin \mathrm{CI}_\zeta (\tau_k (\tilde{t}))$ where $\mathrm{CI}_\zeta (\cdot)$ denotes the $\zeta . 100\%$ confidence interval of BLR prediction with $\zeta \in (0,1)$, and $\tilde{t}_k$ is the time that the last prediction for $\tau_k$ is stored.
\end{definition}

Definition~\ref{def:replanning} implies that replanning is activated at time $t^{\rm{c}}$ if there is an HDV, \eg \HDV{k}, whose actual time shift at $t^{\rm{c}}$ is outside the $\zeta . 100\%$ confidence interval of the last stored prediction.
The time that the last prediction for $\tau_k$ is stored can be either the entry time of \HDV{k} or the previous replanning instance.
Once replanning is activated, we retrain the BLR model, update the trajectory and merging time predictions for the HDVs, and resolve Problem~\ref{prb:cc-optimal_MZ} given new initial conditions for some specific CAVs.
The set of CAVs that need replanning is given in the following definition. 
\begin{definition}
\label{def:replan-set}
At a replanning instant $t^{\rm c}$, let 
\[\HHH' (t^{\rm c}) := \{ k \in \HHH (t^{\rm c}) \,|\, \hat{\tau}_k (t^{\rm{c}}) \notin \mathrm{CI}_\zeta (\tau_k (\tilde{t}))\}\] 
be set of all HDVs that violate the condition ${\hat{\tau}_k (t^{\rm{c}}) \in \mathrm{CI}_\zeta (\tau_k (\tilde{t}))}$.
Let \HDV{j} be the HDV with the minimum predicted merging time in $\HHH' (t^{\rm c})$, \ie ${\mu_{t_j^{\rm m}} \le \mu_{t_k^{\rm m}}, \, \forall k \in \HHH' (t^{\rm c})}$.
The set of CAVs that need replanning is determined as follows
\begin{equation}
\begin{multlined}
\AAA' (t^{\rm c}) := \{ i \in \AAA (t^{\rm c}) \cap \RRR_{j,\rm S} (t^{\rm c}) \,|\, t_i^{0} > t_j^{0}\} \\ 
\cup 
\{ i \in \AAA (t^{\rm c}) \cap \RRR_{j,\rm N} (t^{\rm c}) \,|\, t_i^{\rm m} > \mu_{t_j^{\rm m}} - \rho_l \},
\end{multlined}
\end{equation}
where $t_i^{\rm m}$ is the planned merging time for \CAV{i}.
\end{definition}
Definition~\ref{def:replan-set} means that we replan for \CAV{i} either if (1) \CAV{i} travels on the same road to \HDV{j} and enters the control zone after \HDV{j} or (2) \CAV{i} travels on the neighboring road to \HDV{j} and the planned merging time is greater than $\mu_{t_j^{\rm m}} - \rho_l$.
The stochastic time-optimal control problem at any time $t^{\rm c}$ when a replanning event is activated can be given as follows.
\begin{problem} 
(\textbf{Stochastic replanning in the control zone})
At the time $t^{\rm c}$ with an replanning event, \CAV{i} $\in \AAA' (t^{\rm c})$ solves the following time-optimal control problem
\label{prb:cc-replan}
\begin{align}
\begin{split}
 &\minimize_{t_i^{\rm f} \in \mathcal{T}_i(t^{\rm c})} \quad t_i^{\rm f} 
 \\
 & \subjectto 
 \\
 & \quad \eqref{eq:uconstraint}, \eqref{eq:vconstraint}, \eqref{eq:optimalTrajectory},
 \\
 & \quad \eqref{eq:cc-lateral},\; \forall \, k \in \RRR_{i,\mathrm{N}}(t^{\rm c}), 
 \\
 & \quad \eqref{eq:cc-rearend-1},\; k = \max \, \{ j \in \RRR_{i,\mathrm{S}}(t^{\rm c}) \; | \; j < i \}, 
 \\
 & \quad \eqref{eq:cc-rearend-2},\; k = \max \, \{ j \in \RRR_{i,\mathrm{N}}(t^{\rm c}) \; | \; t_j^{\rm m} < t_i^{\rm m} \},
 \\
 & \text{given:} 
 \\
 & \quad p_i (t^{\rm c}), \,\, v_i (t^{\rm c}), \\
 & \quad p_i (t_i^{\rm f}) = p^{\rm f}, \,\, u_i (t_i^{\rm f}) = 0. 
\end{split}
\end{align}
\end{problem}

The replanning mechanism is thus summarized in Algorithm~\ref{alg:replan}.

\section{Simulation Results}
\label{sec:sim}

In this section, we demonstrate the control performance of the proposed framework by numerical simulations.

\subsection{Simulation Setup}

\begin{table}[tb!]
\caption{Parameters of the trajectory planning framework.}
\label{tab:parameters} 
\centering
\begin{tabular}{ p{0.09\textwidth} p{0.10\textwidth} p{0.09\textwidth} p{0.10\textwidth} }
\toprule[1pt]%
\textbf{Parameters} & \textbf{Values} & \textbf{Parameters} & \textbf{Values} \\
\midrule[0.5pt] %
$v_{\max}$ & $\SI{30.0} {m/s}$ & $v_{\min}$ & $\SI{3.0} {m/s}$ \\
$a_{\max}$ & $\SI{3.0} {m/s^2}$ & $a_{\min}$ & $\SI{-4.0} {m/s^2}$ \\
$\rho_{\rm l}$ & $\SI{2.5}{s}$ & $\rho_{\rm r}$ & $\SI{1.5}{s}$ \\
$d_{\min}$ & $\SI{10}{m}$ & $H$ & $20$ \\
$\xi$ & $0.95$ & $\zeta$ & $0.8$ \\
\bottomrule[1pt] %
\end{tabular}
\end{table}

For our simulation, we used PTV\textendash VISSIM \cite{evanson2017connected} which is a commercial software for simulating microscopic multimodal traffic flow. 
PTV\textendash VISSIM provides a human-driven psycho-physical perception model created by Wiedemann \cite{wiedemann1974simulation}. 
To emulate the behavior of human drivers in an unsignalized merging scenario, we leveraged the network object called ``conflict areas'' of the software where we assigned undetermined priority for the vehicles moving on two roadways.
In the simulation, we considered a merging scenario with a buffer zone of length \SI{70}{m}, a control zone of length \SI{430}{m} (\SI{350}{m} upstream and \SI{80}{m} downstream of the merging point), and a virtual projection zone of length \SI{100}{m}.
The simulation environment in PTV\textendash VISSIM is shown in Fig~\ref{fig:vissim}.
The proposed trajectory planning framework was implemented using the Python programming language with the parameters given in Table~\ref{tab:parameters}.
In addition, the speed of congestion wave $w$ in Newell's car-following model was chosen based on the traffic volume such that $3600\, w/n = 10$, where $n$ denotes the traffic volume in vehicles per hour, which implies that the average standstill spacing is $\SI{10}{m}$.
Videos and data of the simulations can be found at \url{https://sites.google.com/cornell.edu/tcst-cav-mt}.

\subsection{Results and Discussions}

\begin{figure}[!t]
 \begin{center}
 \includegraphics[scale=0.22]{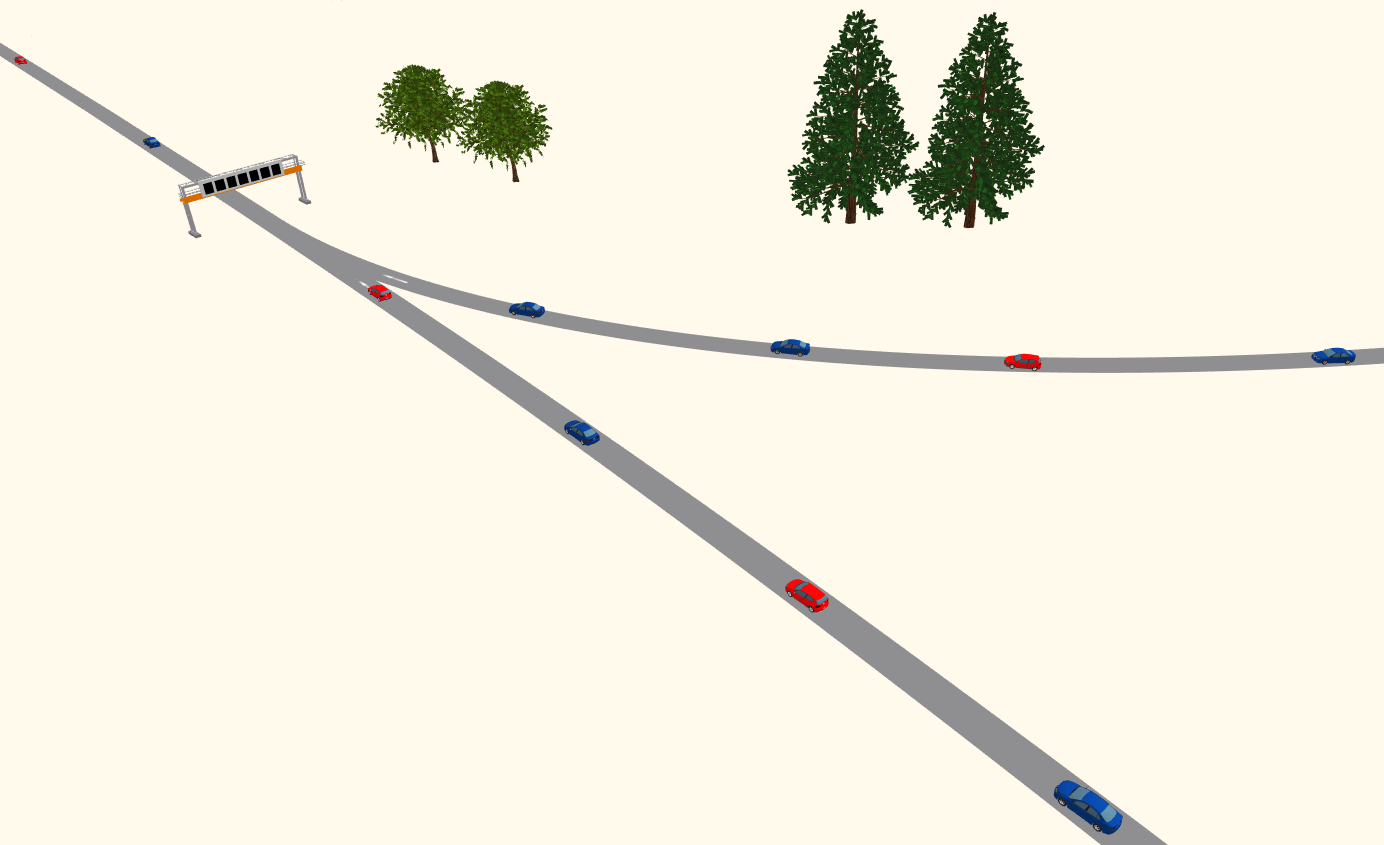}
 \end{center}
\caption{The simulation environment in PTV\textendash VISSIM.} 
\label{fig:vissim}
\end{figure}

\begin{figure*}[!t]
\hspace{-7.5pt}
\begin{subfigure}[t]{.245\textwidth}
\includegraphics[scale=0.29]{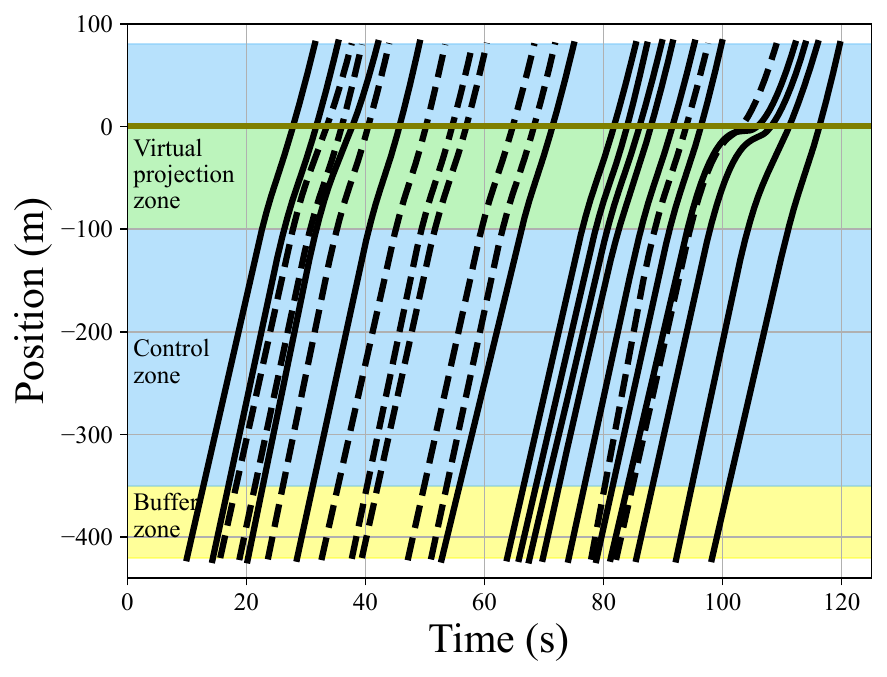}
\caption{$0\%$ penetration rate}
\label{fig:sim_pene-a}
\end{subfigure}
\begin{subfigure}[t]{.245\textwidth}
\includegraphics[scale=0.29]{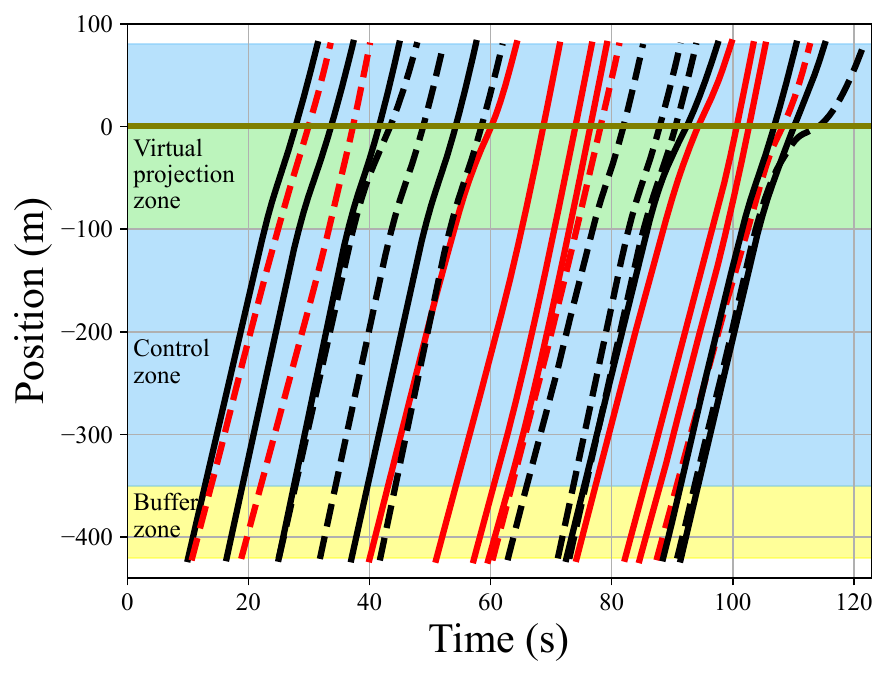}
\caption{$40\%$ penetration rate}
\label{fig:sim_pene-b}
\end{subfigure}
\begin{subfigure}[t]{.245\textwidth}
\includegraphics[scale=0.29]{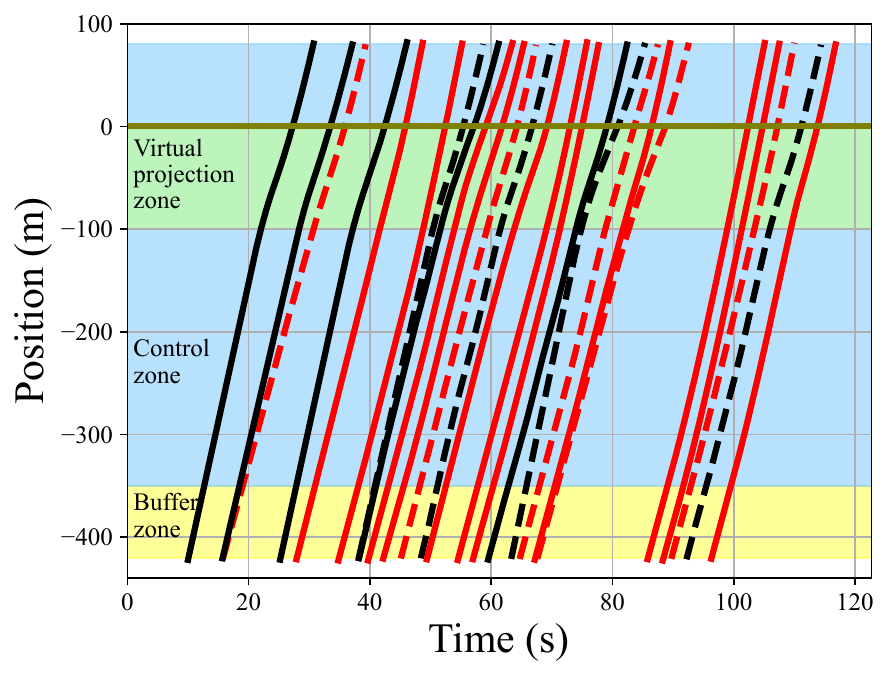}
\caption{$60\%$ penetration rate}
\label{fig:sim_pene-c}
\end{subfigure}
\begin{subfigure}[t]{.245\textwidth}
\includegraphics[scale=0.29]{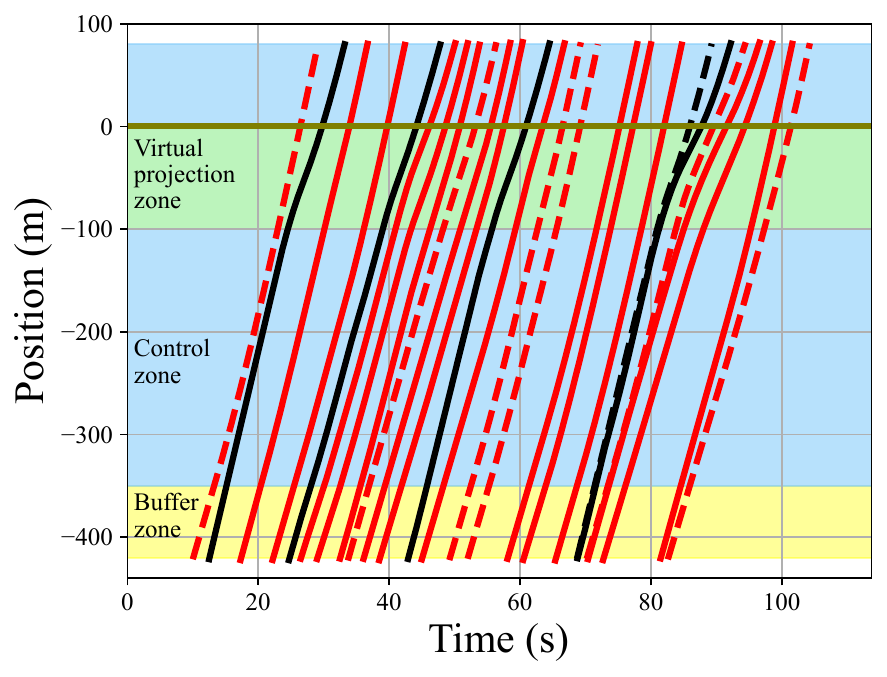}
\caption{$80\%$ penetration rate}
\label{fig:sim_pene-d}
\end{subfigure}

\begin{subfigure}[t]{.245\textwidth}
\includegraphics[scale=0.29]{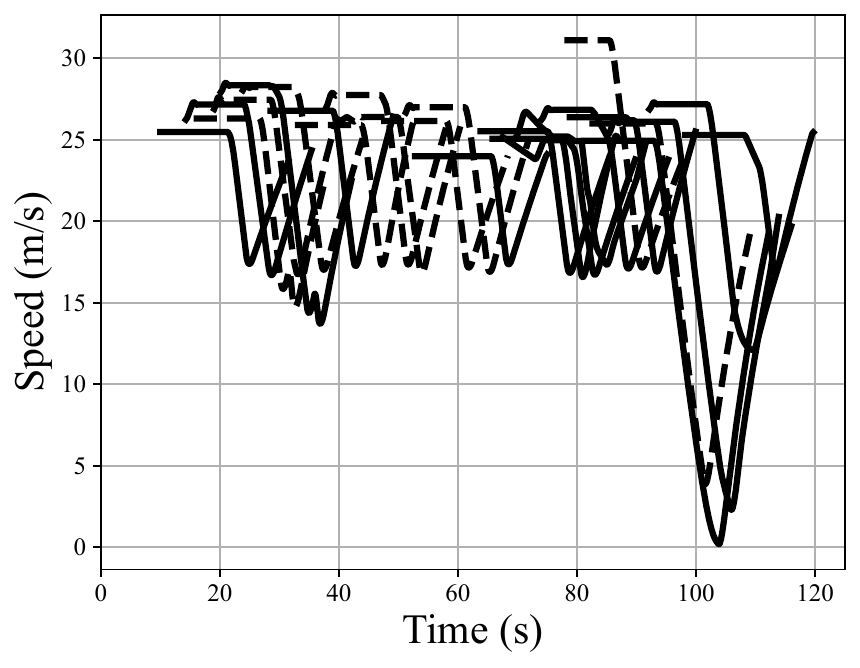}
\caption{$0\%$ penetration rate}\
\label{fig:sim_pene-e}
\end{subfigure}
\begin{subfigure}[t]{.245\textwidth}
\includegraphics[scale=0.29]{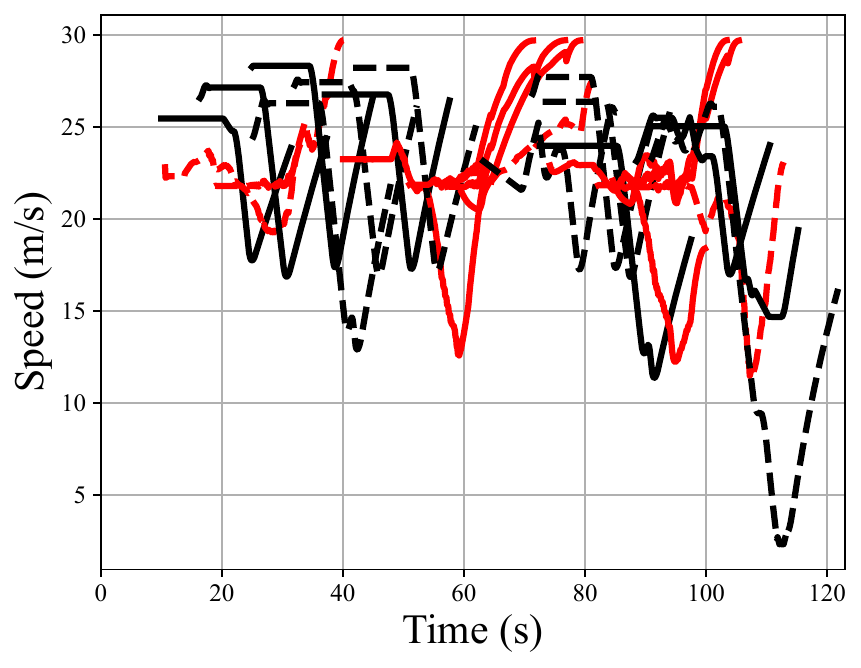}
\caption{$40\%$ penetration rate}
\label{fig:sim_pene-f}
\end{subfigure}
\begin{subfigure}[t]{.245\textwidth}
\includegraphics[scale=0.29]{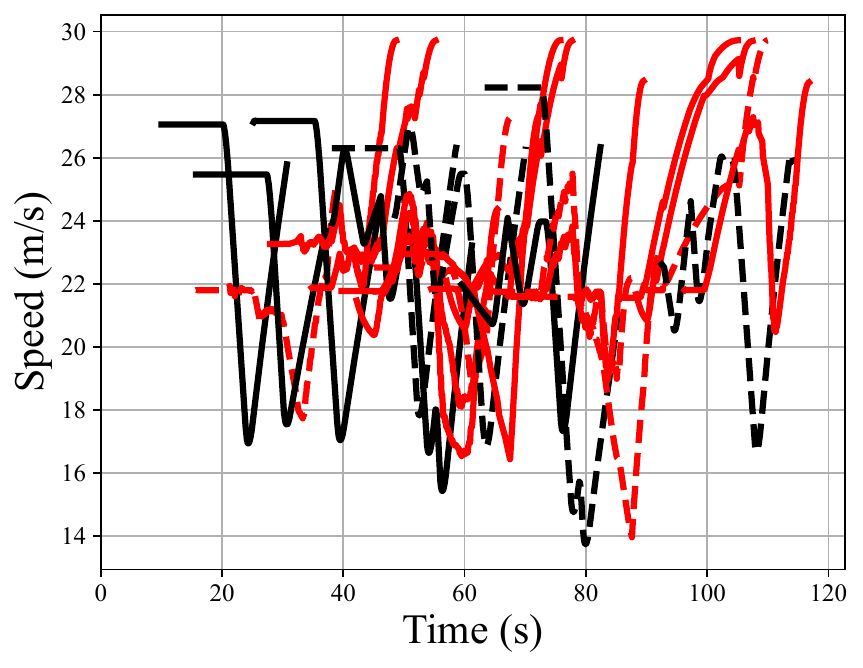}
\caption{$60\%$ penetration rate}
\label{fig:sim_pene-g}
\end{subfigure}
\begin{subfigure}[t]{.245\textwidth}
\includegraphics[scale=0.29]{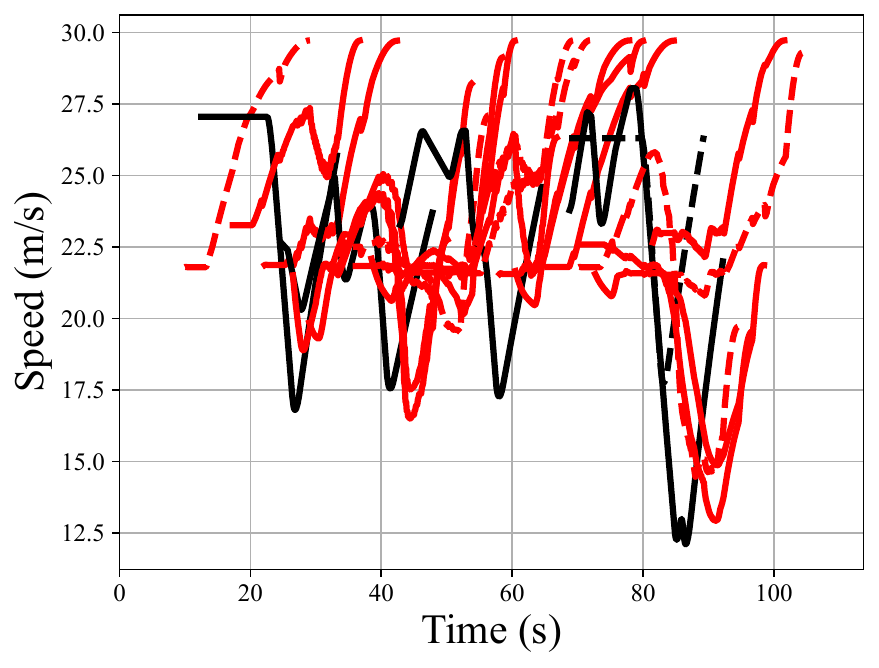}
\caption{$80\%$ penetration rate}
\label{fig:sim_pene-h}
\end{subfigure}
\caption{Position trajectories and speed profiles of the first 25 vehicles in four simulations with different penetration rates. 
The trajectories for CAVs and HDVs are represented by red and black curves, respectively.
The vehicles moving on different roads are distinguished by solid curves and dashed curves.} 
\label{fig:sim_pene}
\end{figure*}

\begin{table}
\centering
\caption{Average travel time (in seconds) under different penetration rates and traffic volumes.}
\label{tab:traveltime}
\begin{tabular}{ K{3.25cm} | K{0.6cm} K{0.6cm} K{0.6cm} K{0.6cm} K{0.6cm} }
\toprule[1.2pt]%
\backslashbox{Traffic\\volume}{Penetration\\rate} & $0\%$ & $40\%$ & $60\%$ & $80\%$ & $100\%$ \\
\midrule[0.6pt]
$800$ (veh/h) & $23.9$ & $22.4$ & $22.3$ & $21.2$ & $20.4$ \\ %
$1000$ (veh/h) & $25.6$ & $23.4$ & $22.7$ & $21.6$ & $21.0$ \\ %
$1200$ (veh/h) & $29.1$ & $25.0$ & $24.3$ & $23.1$ & $21.7$ \\
\bottomrule[1.2pt]
\end{tabular}
\end{table}

We conducted multiple simulations for three traffic volumes: 800, 1000, and 1200 vehicles per hour along with five different penetration rates: $0\%$, $40\%$, $60\%$, $80\%$, and $100\%$.
In each simulation, we collected data for $500$ seconds to compute the average travel time of the vehicles and reported the results in Table~\ref{tab:traveltime}.
As can be seen from the table, at higher penetration rates, average travel times significantly improve compared to baseline traffic consisting solely of HDVs across all tested traffic volumes.
For example, in the simulation with a high traffic volume of $1200$ vehicles per hour, $40\%$, $60\%$, $80\%$, and $100\%$ penetration rates can reduce average travel time by $14.1\%$, $16.5\%$, $20.6\%$, and $25.4\%$, respectively.
The results also suggest that high penetration rates may be necessary for enhancing mixed traffic under high-volume conditions.
Next, we show the position trajectories and speed profiles of the first 25 vehicles in four simulations under $0\%$, $40\%$, $60\%$, and $80\%$ penetration rates and the traffic volume of $1200$ vehicles per hour in Fig.~\ref{fig:sim_pene}.
The trajectories and speed profiles for $100\%$ CAV coordination are similar to previous studies, \eg \cite{mahbub2020decentralized}, and are thus omitted.
Overall, the results show that under partial penetration rates, \ie $40\%$, $60\%$, and $80\%$, the proposed framework guarantees safe co-existence between CAVs and HDVs (\cf Figures~\ref{fig:sim_pene-b}, \ref{fig:sim_pene-c}, and \ref{fig:sim_pene-d}).
Moreover, Figures~\ref{fig:sim_pene}(e\textendash h) suggest the potential benefits of coordination under increased CAV penetration rates in reducing traffic disruption.
It is observed that HDVs generally exhibit more abrupt deceleration and acceleration compared to CAVs.

To better illustrate the advantages of the replanning mechanism, we show in Fig.~\ref{fig:sim_replan} the position trajectories of some vehicles in the simulation with $60\%$ penetration rate where without replanning the safety constraints are violated.
The top panels of Fig.~\ref{fig:sim_replan} reveal that the optimal trajectory of the CAVs, derived at the entry of the control zone, may cause a collision with either the preceding HDV or the HDV entering from the neighboring road due to the discrepancy between the HDV’s predicted trajectory and the actual trajectory.
On the other hand, the bottom panels demonstrate that with the proposed replanning mechanism, the CAVs are able to detect the changes in human driving behavior and replan a new trajectory to avoid collisions with the HDVs. 

Safe maneuvers for CAVs can be further enhanced by using probabilistic constraints. 
In Fig.~\ref{fig:sim_cc}, we show the deterministic and robust trajectories derived at the control zone entry for particular CAVs in two simulations.
For comparison purposes, we do not consider replanning in those simulations.
In the first simulation (Fig.~\ref{fig:sim_cc-a}), we define the unsafe region for merging time that is determined by values at which the tightened lateral constraint \eqref{eq:cc-lateral-a} is violated. 
Likewise, in the second simulation (Fig.~\ref{fig:sim_cc}-b), from the distribution of the time shift prediction, we compute the unsafe region where the tightened rear-end constraint \eqref{eq:cc-rearend-1a} is violated.
We can observe that in both cases the robust trajectory can ensure that the trajectory and merging time of the CAV do not invade the unsafe regions.
Conversely, the deterministic trajectory violates the unsafe regions and may result in slightly more aggressive behavior.
Note that the cautiousness of the stochastic planning framework can be adjusted by changing the probability of constraint satisfaction.

\begin{figure}[!t]
\hspace{-5pt}
\begin{subfigure}{.49\textwidth}
\includegraphics[scale=0.29]{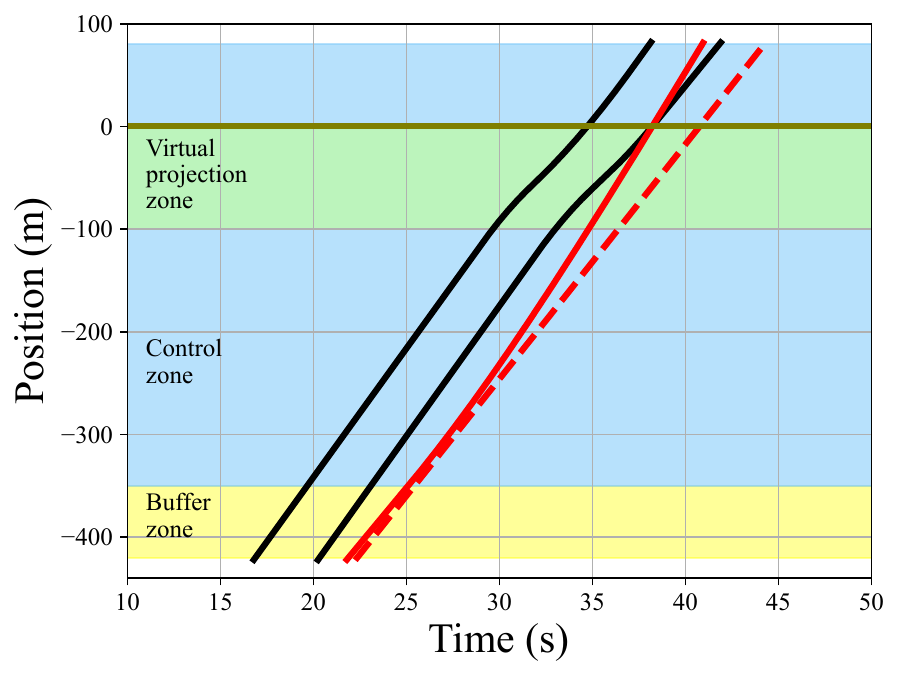}
\includegraphics[scale=0.29]{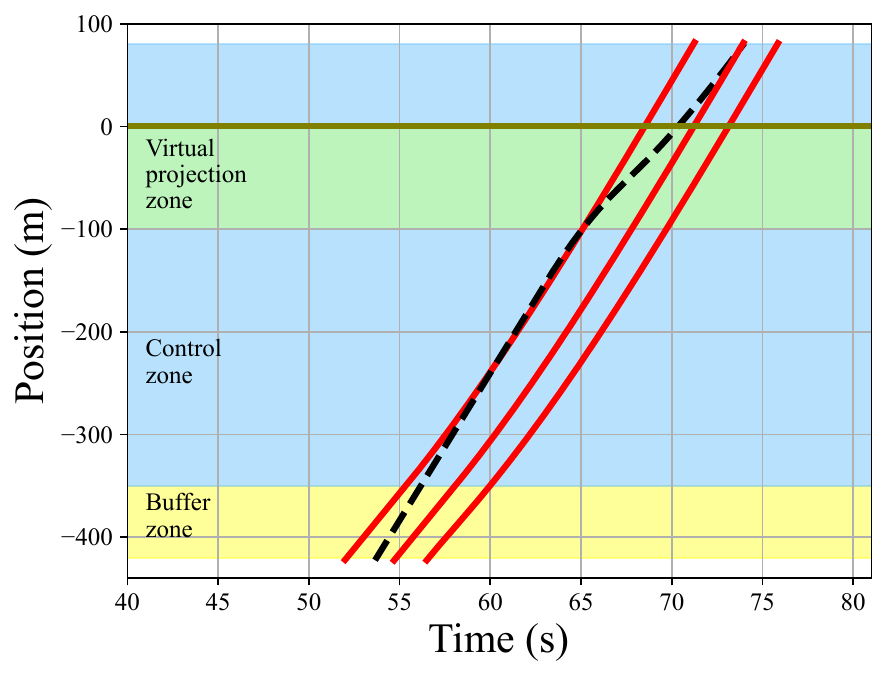}
\caption{Without replanning}
\label{fig:sim_replan-a}
\end{subfigure}

\hspace{-5pt}
\begin{subfigure}{.49\textwidth}
\includegraphics[scale=0.29]{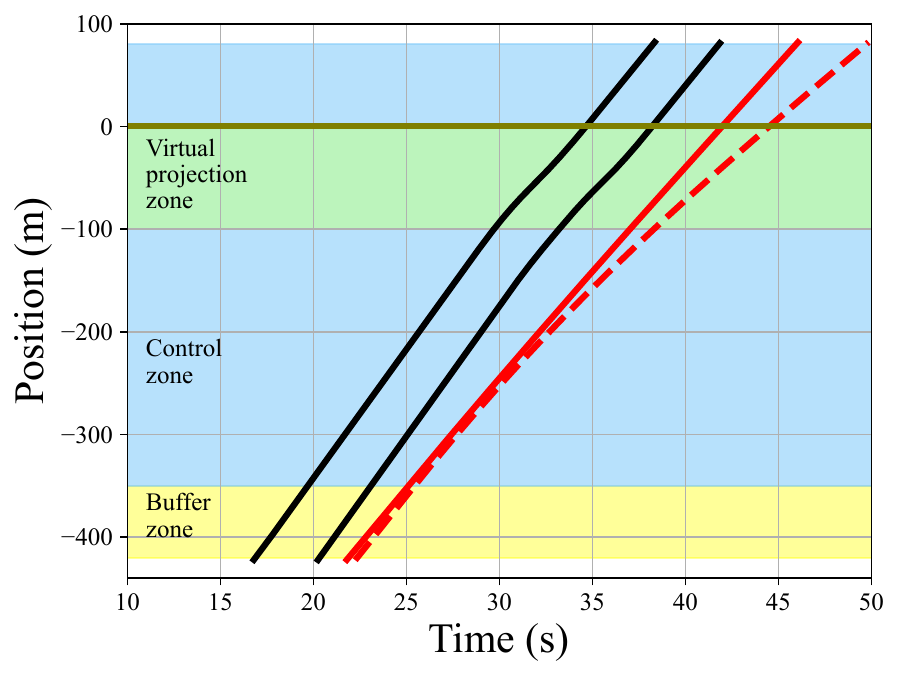}
\includegraphics[scale=0.29]{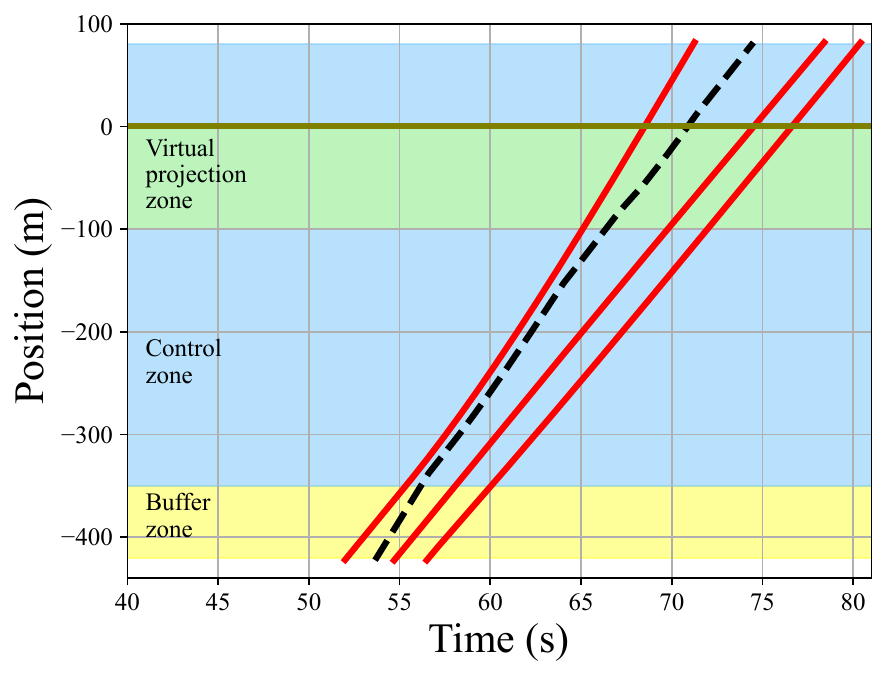}
\caption{With replanning}
\label{fig:sim_replan-d}
\end{subfigure}

\caption{Position trajectories for CAVs (red) and HDVs (black) in simulations without (a) and with replanning (b). In the simulations without replanning, the safety constraints are activated for CAVs and HDVs entering from the same road (a-left) and from the neighboring road (a-right).
The vehicles moving on different roads are distinguished by solid curves and dashed curves.} 
\label{fig:sim_replan}
\end{figure}

\begin{figure}[!t]
\hspace{-5pt}
\begin{subfigure}{.245\textwidth}
\includegraphics[scale=0.29]{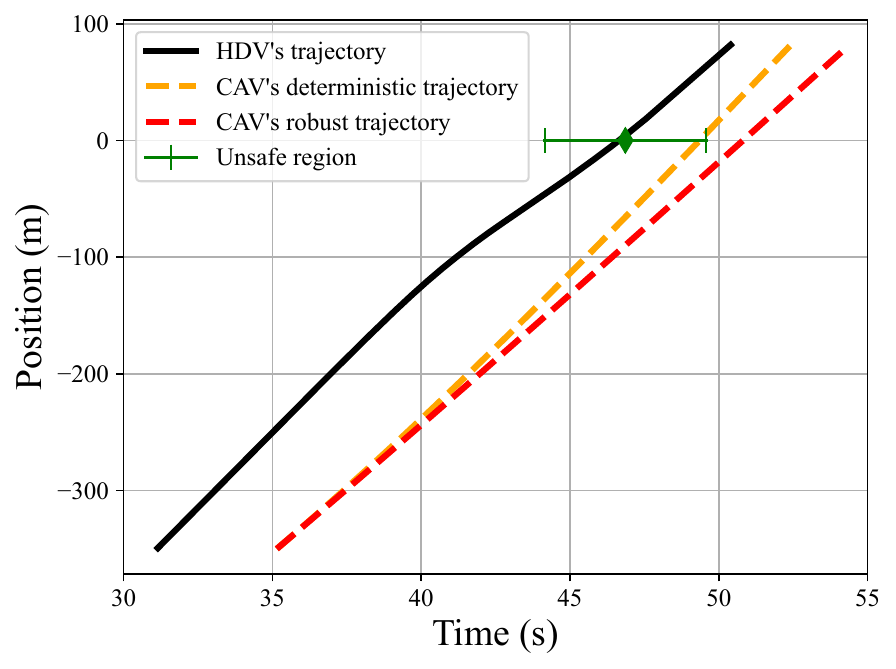}
\caption{Lateral constraint}
\label{fig:sim_cc-a}
\end{subfigure}
\hspace{-5pt}
\begin{subfigure}{.245\textwidth}
\includegraphics[scale=0.29]{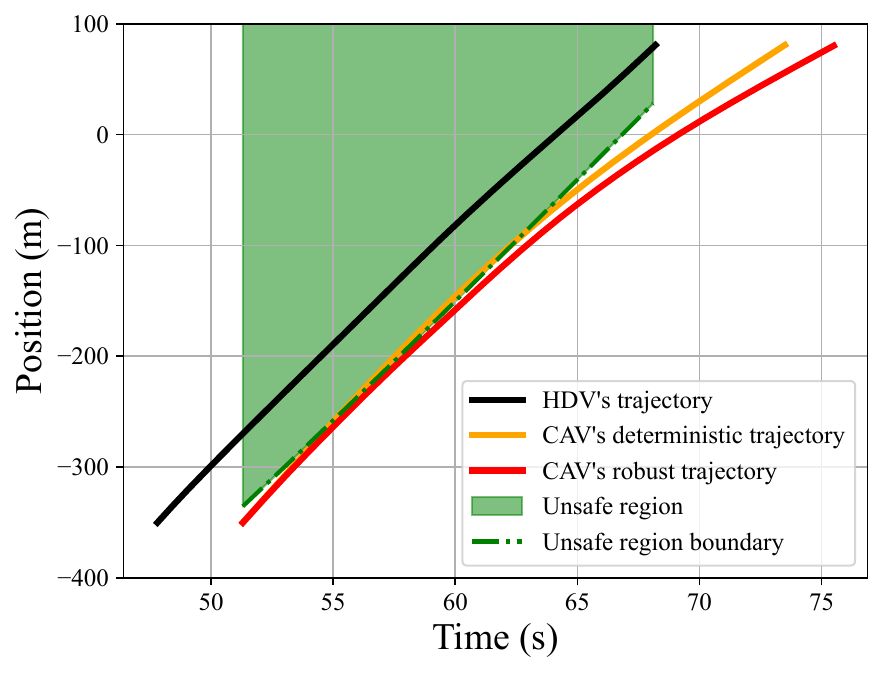}
\caption{Rear-end constraint}
\label{fig:sim_cc-b}
\end{subfigure}
\caption{Comparison between deterministic and robust trajectories for CAVs.
The vehicles moving on different roads are distinguished by solid curves and dashed curves.} 
\label{fig:sim_cc}
\end{figure}

\section{Conclusions}
\label{sec:conc}

In this paper, we presented a stochastic time-optimal trajectory planning framework for CAVs in mixed-traffic merging scenarios. 
We proposed a data-driven Newell's car-following model in which the time shift is calibrated online using Bayesian linear regression for modeling human driving behavior and the virtual projection technique is used to capture the lateral interaction.
We applied the data-driven Newell's car-following model to predict the trajectories and merging times of HDVs along with quantifying the prediction uncertainties used for probabilistic constraints in the stochastic time-optimal control problem.
Finally, we developed a replanning mechanism to activate resolving the stochastic time-optimal control problem for CAVs if the last stored predictions are not sufficiently accurate compared to the actual observations.
The results from simulations validate that our proposed framework can ensure safe maneuvers for CAVs among HDVs, and under higher penetration rates the mixed traffic can be improved to some extent. 

There are several research directions that can be considered in our future work.
First, we will focus on extending the framework to consider more challenging scenarios such as multi-lane merges and intersections.
Additionally, since the interaction between CAVs and HDVs becomes more complex and the coordination framework's efficiency diminishes given high traffic volumes, the ideas from optimal routing \cite{Bang2022combined,bang2023routing} can be combined to control the traffic flow. 
Finally, we plan to validate the proposed framework in an experimental robotic testbed where human participants can drive robotic vehicles to constitute realistic mixed traffic \cite{chalaki2021CSM}.

\section*{Acknowledgments}
The authors would like to thank Dr. Ehsan Moradi-Pari at Honda Research Institute USA, Inc. for his valuable feedback on the manuscript.

\bibliographystyle{IEEEtran}
\bibliography{IEEEabrv,bib/refs_IDS,bib/refs_ctrl,bib/refs_cav,bib/refs_ml,bib/refs_others,bib/refs_mt,bib/my_pub}

\appendices
\section{Proof of Lemma 1}
\label{sec:lem1}
\begin{proof}
Note that $\lambda_k \sim \NNN (\lambda_k, \sigma^2_{\tau_k})$ where ${\lambda_k = t- \mu_{\tau_k}}$.
From Newell's car-following model, we have
\begin{equation} 
\begin{split}
p_k(t) &= p_j(t - \tau_k (t)) - w \tau_k\\
&
\begin{multlined} \label{eq:p_k}
= \phi_{j,3} \lambda_k^3 + \phi_{j,2} \lambda_k^2 + (\phi_{j,1} + w)\lambda_k + (\phi_{j,0} - w t).
\end{multlined}
\end{split}
\end{equation}
The predicted mean of $p_k(t)$ can be found by taking the expectation of \eqref{eq:p_k}, i.e.,
\begin{equation}
 \EE[p_k(t)]= \EE[\phi_{j,3} \lambda_k^3 + \phi_{j,2} \lambda_k^2 + (\phi_{j,1} + w)\lambda_k + (\phi_{j,0} - w t)].
\end{equation}
From the linearity of expectation, we have 
\begin{equation} \label{eq:Expp_k_Step1}
 \EE[p_k(t)]= \phi_{j,3}\EE[ \lambda_k^3] + \phi_{j,2} \EE[\lambda_k^2] + (\phi_{j,1} + w)\EE[\lambda_k] + (\phi_{j,0} - w t),
\end{equation}
where $\EE[ \lambda_k^n]$ denotes the $n^\text{th}$ moment of random variable $\lambda_k \sim \NNN (\lambda_k, \sigma^2_{\tau_k})$. The $n^\text{th}$ moment of random variable $\lambda_k $ can be obtained by evaluating the $n^\text{th}$ derivative of moment-generating function $M_{\lambda}$ with respect to the slack variable $\tau$ and setting $\tau$ equal to zero, namely, $\EE[ \lambda_k^n]=\dfrac{d^n}{d\tau^n} M_{\lambda}(\tau)\mid_{\tau=0}$. The moment-generating function of random variable $\lambda_k $ which is taken from a normal distribution is given by $M_e(\tau)=\exp(\tau \lambda_k + \dfrac{1}{2}\sigma_{\tau_k}^2\tau^2)$, where $\lambda_k$ and $\sigma_{\tau_k}^2$ denote the mean and variance of the random variable, respectively. Following the above process, the first, second, and third moments of $\lambda_k$ are derived as follows
\begin{align}
 \EE[ \lambda_k^3] &= \lambda_k^3 + 3\lambda_k \sigma_{\tau_k}^2, \label{eq:3moment} \\ 
 \EE[ \lambda_k^2] &= \lambda_k^2 + \sigma_{\tau_k}^2, \label{eq:2moment}\\ 
 \EE[ \lambda_k] &= \lambda_k. \label{eq:1moment}
\end{align}

Substituting \eqref{eq:3moment}-\eqref{eq:1moment} in \eqref{eq:Expp_k_Step1}, we get the the predicted mean of $p_k(t)$ as derived in \eqref{eq:mu-f}.
To find variance of $p_k(t)$, we employ $\sigma^2_{p_k}(t) = \EE[p_k(t)^2] -\EE[p_k(t)]^2$. 
The computation of the second term, $\EE[p_k(t)]^2$, can be obtained by squaring the \eqref{eq:mu-f}. However, to derive the first term, $\EE[p_k(t)^2]$, it is necessary first to compute $p_k(t)^2$, and afterward take its expectation. 
Utilizing \eqref{eq:p_k}, we have
\begin{equation} \label{eq:p_k_Square}
p_k(t)^2= (\phi_{j,3} \lambda_k^3 + \phi_{j,2} \lambda_k^2 + \phi_{j,1}^{\prime} \lambda_k + \phi_{j,0}^{\prime})^2,
\end{equation}
where, $\phi_{j,1}^{\prime}=\phi_{j,1} + w$ and $ \phi_{j,0}^{\prime}=\phi_{j,0} - w t$.
Expanding \eqref{eq:p_k_Square}, we have 
\begin{multline} \label{eq:p_k_Squared_Expand}
 p_k(t)^2= \lambda_k^6 \,{\phi_{j,3} }^2 +2\,\lambda_k^5 \,\phi_{j,2} \,\phi_{j,3} +2\,\lambda_k^4 \,\phi_{j,1}^{\prime} \,\phi_{j,3} + \lambda_k^4 \,{\phi_{j,2} }^2 \\
 +2\,\lambda_k^3 \,\phi_{j,0}^{\prime} \,\phi_{j,3} +2\,\lambda_k^3 \,\phi_{j,1}^{\prime} \,\phi_{j,2} +2\,\lambda_k^2 \,\phi_{j,0}^{\prime} \,\phi_{j,2} +\lambda_k^2 \,{\phi_{j,1}^{\prime} }^2 \\+2\,\lambda_k \,\phi_{j,0}^{\prime} \,\phi_{j,1}^{\prime} +{\phi_{j,0}^{\prime} }^2.
\end{multline}
To compute expectation of \eqref{eq:p_k_Squared_Expand}, we first need to derive fourth, fifth, and sixth moment of random variable $\lambda_k$ as
\begin{align}
 \EE[ \lambda_k^6] &= \lambda_k^6 +15\lambda_k^4 \,\sigma_{\tau_k}^2 +45\,\lambda_k^2 \,\sigma_{\tau_k}^4 +15\,\sigma_{\tau_k}^6\label{eq:6moment} \\ 
 \EE[ \lambda_k^5] &= \lambda_k^5 +10\,\lambda_k^3 \,\sigma_{\tau_k}^2 +15\,\lambda_k \,\sigma_{\tau_k}^4, \label{eq:5moment}\\ 
 \EE[ \lambda_k^4] &= \lambda_k^4+6\lambda_k^2\sigma_{\tau_k}^2+3\sigma_{\tau_k}^4. \label{eq:4moment}
\end{align}
Next, we can derive the second moment of random variable $p_k(t)$ i.e., $\EE[p_k(t)^2]$, by taking the expectation of \eqref{eq:p_k_Squared_Expand} using the linearity of expectation and \eqref{eq:3moment}-\eqref{eq:1moment} and \eqref{eq:6moment}-\eqref{eq:4moment}. Substituting the results in $\sigma^2_{p_k}(t) = \EE[p_k(t)^2] -\EE[p_k(t)]^2$ , and performing some simple algebraic manipulations, we obtain \eqref{eq:sigma-f}. 
The derivation of $\bbsym{\phi}_k$ in \eqref{eq:hdv-poly} results from equating the coefficients of two polynomials in the left-hand and right-hand sides of \eqref{eq:mu-f}, and the proof is complete.
\end{proof}

\end{document}

%% file: defs.tex
\newcommand{\probP}{\text{I\kern-0.15em P}}

\newcommand{\BLR}{\ensuremath{\BBB}}